\documentclass[aip,jcp,10pt,nofootinbib,twocolumn]{revtex4-1}
\usepackage{amsmath,amssymb,amsfonts,amsthm}
\usepackage{bbm}
\usepackage{mathrsfs}
\usepackage{algorithmic}
\usepackage{graphicx}

\providecommand{\ud}{\,\mathrm{d}}
\providecommand{\abs}[1]{\left\lvert#1\right\rvert}
\providecommand{\bra}[1]{\left\langle#1\right\vert}
\providecommand{\ket}[1]{\left\vert#1\right\rangle}
\providecommand{\hprod}[2]{\left\langle#1\,\vert\,#2\right\rangle}
\providecommand{\bprod}[2]{\left(#1\,\vert\,#2\right)}
\providecommand{\vect}[1]{{\boldsymbol{#1}}}
\providecommand{\norm}[1]{\left\lVert#1\right\rVert}
\providecommand{\suchthat}{\,:\,}
\providecommand{\evalat}{\Big\vert}
\providecommand{\conj}[1]{\overline{#1}}
\providecommand{\trace}{\mathrm{tr}}

\providecommand{\R}{{\mathbb R}}

\providecommand{\Z}{{\mathbb Z}}
\providecommand{\N}{{\mathbb N}}
\providecommand{\ii}{\mathbbm{i}}
\providecommand{\basespacet}{\R^3\times\{-\frac12,+\frac12\}}
\providecommand{\basespacem}{\R^3\times\left\{-\mbox{$\frac12$},+\mbox{$\frac12$}\right\}}

\newtheorem{definition}{Definition}
\newtheorem{theorem}[definition]{Theorem}
\newtheorem{proposition}[definition]{Proposition}

\begin{document}

\title{Efficient Algorithm for Asymptotics-Based Configuration-Interaction Methods and Electronic Structure of Transition Metal Atoms}

\author{Christian B. Mendl$^1$, Gero Friesecke$^1$}
\affiliation{$^1$Center for Mathematics, TU Munich}

\date{\today}

\begin{abstract}
\small
Asymptotics-based configuration-interaction (CI) methods [G.~Friesecke and B.~D.~Goddard,
Multiscale Model. Simul.~\textbf{7}, 1876 (2009)]\nocite{CI2009} are a class of CI methods for atoms which reproduce,
at fixed finite subspace dimension, the exact Schr\"odinger eigenstates in the limit of fixed electron number
and large nuclear charge. Here we develop, implement, and apply to $3d$~transition metal atoms an efficient and
accurate algorithm for asymptotics-based CI.

Efficiency gains come from exact (symbolic) decomposition of the CI space into irreducible symmetry subspaces
at essentially linear computational cost in the number of radial subshells
with fixed angular momentum, use of reduced density matrices in order to avoid having to store
wavefunctions, and use of Slater-type orbitals (STO's).
The required Coulomb integrals for STO's are evaluated in closed form, with the help of Hankel matrices, Fourier analysis, 
and residue calculus.

Applications to $3d$~transition metal atoms are in good agreement with experimental
data. In particular we reproduce the anomalous magnetic moment and orbital filling
of Chromium in the otherwise regular series Ca, Sc, Ti, V, Cr.
\end{abstract}


\maketitle

\section{Introduction}
\label{sec:Introduction}

The search for accurate computational methods for the $N$-electron Schr\"odinger equation at moderate computational 
cost has been a focus of activity for several decades\cite{KohnSham65,SzaboOstlund,HJO00,Mazziotti2007}. The present
article is a contribution to one part of the picture, wavefunction methods for atoms. We develop, implement, and apply
to transition metal atoms an algorithmic framework which renders asymptotics-based Configuration-Interaction~(CI)
computations for atoms with basis sets of up to $50$ one-electron spin orbitals, up to $30$ electrons, and full resolution
of all valence electron correlations feasible.
An attractive feature of our framework is that many steps are done \emph{symbolically}, by building upon, 
systematizing, and automatizing the paper-and-pencil analysis of asymptotics based CI for small atoms and minimal 
bases in Ref.~\onlinecite{CI2009}. A Matlab/Mathematica implementation is available at Ref.~\onlinecite{FermiFab}.

CI methods\cite{Lowdin1955,HJO00} approximate the electronic Schr\"odinger equation by projecting it onto a well 
chosen subspace spanned by Slater determinants. More precisely, the Schr\"odinger equation for an atom or ion with 
$N$ electrons is
\begin{equation}
H \psi = E \psi, \qquad \psi \in L_a^2\left(\left(\basespacem\right)^N\right)
\label{eq:Schroedinger}
\end{equation}
where $H$ is the Hamiltonian of the system, see~\eqref{eq:Hamiltonian} below, $\psi$ the wavefunction and $E$ the energy. 
The wavefunction $\psi = \psi(x_1,s_1,\dots,x_N,s_N)$ depends on the positions $\vect{x}_i\in\R^3$ and 
spins $s_i\in\{-\frac12,+\frac12\}$ of all electrons, and belongs to the space 
$L_a^2((\basespacet)^N)$ of square-integrable, antisymmetric functions on 
$(\basespacet)^N$. A CI method is an approximation of~\eqref{eq:Schroedinger} by an equation of form
\begin{equation}
\begin{split}
P H P\Psi &= E\Psi, \quad \psi \in V\subset L_a^2\left(\left(\basespacem\right)^N\right)\\
P &=\text{ orthogonal projector onto }V,
\end{split}
\label{eq:CI}
\end{equation}
where $V$ is a Span of a finite number of Slater determinants $\ket{\chi_{i_1}\cdots\chi_{i_N}}$ built from a finite 
number of spin orbitals $\{\chi_1,\dots,\chi_K\}\subset L^2(\basespacet)$. We recall the well known
\emph{fundamental difficulty of CI methods}: Eq.~\eqref{eq:Schroedinger} is a partial differential equation in very high space dimension,
e.g. dimension 72 in case of a single Chromium atom as treated in this paper. Hence when discretizing the single-electron state space
by a reasonable number of spin orbitals, $L^2(\basespacet)\approx \mathrm{Span}\,\{\chi_1,..,\chi_K\}$, the ensuing natural choice 
$V=\mathrm{Span}\, \{ \ket{\chi_{i_1}\cdots\chi_{i_N}} \suchthat 1\le i_1<\dots<i_N\le K\}$ (full CI) has a prohibitively
large dimension, $\binom{K}{N}$.

Our principal contribution here is the development of an efficient algorithm that minimizes the curse of dimension. 
The main savings come from exact (i.e. symbolic) and efficiently automated exploitation of symmetry to perform dimension reduction.
Other ingredients are use of reduced density matrices in order to avoid having to store wavefunctions, and use of Slater-type
orbitals (STO's) including exact orthonormalization and Coulomb integral evaluation. The algorithm has been implemented for
a recent variant of CI, asymptotics-based CI\cite{CI2009}, which exploits the asymptotic results in Ref.~\onlinecite{FG2010} and has the
attractive features that the CI subspace, if its dimension
is $K$, reproduces correctly the first $K$ Schr\"odinger eigenstates in the limit of fixed $K$, fixed electron number, and
large nuclear charge $Z$. (This limit, which has a large literature (see in particular Ref.~\onlinecite{Layzer1959,Davidson1996}) captures the physical environment of inner shell electrons, and has the 
\emph{multiscale} property that the ratio of first spectral gap to ground state energy of the Schr\"odinger equation tends to zero\cite{CI2009},
with the experimental ratio for true atoms being very close to zero, about 1 part in 1000 for Carbon and Oxygen
and 1 part in 30 000 for Cr and Fe.) The main part of the algorithm, automated symmetry reduction, can be 
easily adapted to other CI methods and orbitals (such as Gaussians). 
\smallskip

As a typical application, we treat here the transition metal series Ca, Sc, Ti, V, Cr, 
modelled by $18$ core electrons occupying Slater orbitals of type $1s$ to $3p$, and an active space consisting of 
$3d$, $4s$, $4p$, $4d$ Slater orbitals (of either spin) accommodating the two to six valence electrons. The resulting CI space $V$ 
for Cr has dimension $d = \binom{28}{6} = 376 740$, 
and the CI Hamiltonian has $d(d+1)/2\approx 7\times 10^{10}$ entries.
But automated symmetry reduction shows (see Table~\ref{tab:GroundstateEnergies}) that only 14
basis functions contribute to the experimental
ground state configuration and symmetry, [Ar]$4s^13d^5$ $^7S$, allowing to evaluate the ensuing eigenvalues and -states easily
and to machine precision. Our results, detailed in Section~\ref{sec:Applications} below, provide an ab initio explanation 
of the anomalous magnetic moment of Chromium (experimentally, the ground state has six instead of the expected four 
aligned spins) and the underlying anomaly in the filling order of $3d$ versus $4s$ orbitals in the semi-empirical orbital picture of transition metal 
atoms (Chromium, unlike its four predecessors Ca, Sc, Ti, V, possesses only one instead of two $4s$ electrons). It is well
known\cite{MelroseScerri,Kagawa,HarrisJones,Yanagisawa} that single-determinant Hartree-Fock, relativistic Hartree-Fock, and density functional theory calculations 
(even with the best exchange-correlation functionals such as B3LYP) render the correct filling orders and ground state
symmetries only for some but not all transition metal elements (see Section~\ref{sec:Applications}). 
\smallskip

In the remainder of the Introduction we describe our algorithm for \emph{exact (symbolic), efficient symmetry partitioning}.
The (non-relativistic, Born-Oppenheimer) Hamiltonian 
\begin{equation}
H = \sum_{i=1}^N \left(-\frac{1}{2} \Delta_{\vect{x}_i} - \frac{Z}{\abs{\vect{x}_i}}\right) + \sum_{1 \le i < j \le N} \frac{1}{\abs{\vect{x}_i-\vect{x}_j}}
\label{eq:Hamiltonian}
\end{equation}
governing atoms/ions with $N$ electrons and nuclear charge $Z$ has the symmetry group
\begin{equation}
SU(2) \times SO(3) \times \Z_2,
\label{eq:SymGroup}
\end{equation}
consisting of simultaneous rotation of electron spins, and simultaneous rotation and sign reversal of electron positions. This leads to the well known conservation law that the Hamiltonian leaves the simultaneous eigenspaces of the spin, angular momentum and parity operators
\begin{equation}
\vect{L}^2, \ L_z, \ \vect{S}^2, \ S_z, \ \hat{R}
\label{eq:simLS}
\end{equation}
invariant (see Section~\ref{sec:Symmetries} for precise definitions of these operators). 
The fact that partitioning into symmetry subspaces significantly lowers computational costs has long been known to, 
and exploited by, theorists (see e.g. Ref.~\onlinecite{Tayloretal}). 
A striking example is the paper-and-pencil symmetry decomposition\cite{NuclearChargeLimit2009,FG2010}
of a minimal asymptotics-based CI Hamiltonian $P H P$ for the second 
period atoms He to Ne, with active space consisting of the eight $2s$ and $2p$ spin-orbitals accommodating the valence electrons. 
For Carbon, there are four valence electrons, 
so the active space has dimension $\binom{8}{4} = 70$, and the CI Hamiltonian is a $70\times 70$ matrix. But due to symmetry
it decomposes into fifteen $2\times 2$ blocks and fourty $1\times 1$ blocks.

The main algorithmic steps which automate such decompositions are as follows.

(a) One starts by partitioning the CI space into configurations, i.e., subspaces like $1s^{n_1} 2s^{n_2} 2p^{n_3}\dots$ with a fixed number $n_i$ of 
electrons in each subshell (see Section~\ref{sec:FrameworkConfigurations} below).
It suffices to symmetry-decompose each configuration, because the symmetry group, unlike the Hamiltonian, leaves each configuration invariant individually.

(b) Each configuration is isomorphic to a \emph{non-antisymmetrized} tensor product of lower-dimensional factors.
The tensor factors consist of single $1s,\,2s,\,2p,\dots$ subshells. See Section~\ref{sec:FrameworkConfigurations}. 
This product structure is essential for Step (d) below.

(c) The splitting up of each factor into simultaneous eigenspaces of the symmetry operators~\eqref{eq:simLS} is done via a suitable algorithm from the 
mathematics literature for simultaneous diagonalization of commuting matrices, for instance that of Bunse-Gerstnert, Byers and Mehrmann\cite{SimDiag1993}. (We are indebted to Folkmar Bornemann for helpful advice regarding this step.) Exact eigenstates are recovered from the numerical eigenstates through exploiting that the squares of the eigenstate coefficients are, by representation theory, rational numbers.

(d) Given simultaneous eigenstates of the symmetry operators for each factor, simultaneous eigenstates of a two-factor tensor product are known explicitly in terms of the well-known Clebsch-Gordan coefficients, and those for a many-factor tensor product are easily obtained by iteration of the Clebsch-Gordan formulae. This yields the desired decomposition of each configuration.

A key feature of the algorithm (a), (b), (c), (d) is the computational cost grows 
only~\emph{linearly} with the number of subshells, provided the angular momentum cutoff is held fixed. Thus, say, the cost of including $s$ orbitals of 
type $1s,2s,\dots,ns$ is only $\mathcal{O}(n)$. See Section~\ref{sec:Costs}.
\smallskip

The structure of this paper is as follows. In Section~\ref{sec:Framework} we briefly review asymptotics-based CI.
Section~\ref{sec:Reduction} contains the main contribution of this paper, 
namely exact reduction steps leading to significant savings of computational time and memory storage. 
In Section~\ref{sec:Orbitals} treat orthonormalization and Coulomb integral evaluation for general atomic Slater-type orbitals. 
We summarize all algorithmic steps in Section~\ref{sec:Assembly}, and carefully estimate the costs in Section~\ref{sec:Costs}. 
Finally, in the last section we apply the algorithmic framework to the electronic structure of potassium, calcium
and the transition metals scandium to zinc.

\section{Asymptotics-based CI}
\label{sec:Framework}

We briefly recall the set-up and features relevant to the present work, referring to Refs.~\onlinecite{CI2009,NuclearChargeLimit2009} for further information. 

\subsection{Symmetries}
\label{sec:Symmetries}

Due to invariance of the Hamiltonian~\eqref{eq:Hamiltonian} under the symmetry group~\eqref{eq:SymGroup}, 
the set of operators~\eqref{eq:simLS} commutes with the Hamiltonian and with each other, for arbitrary $N$ and $Z$. 
These operators play an important role in our algorithmic framework. 
Here and below we use the standard notation $\vect{L} = \sum_{i=1}^N \vect{L}(i)$ (many-body angular momentum operator), 
$\vect{L}(i) = \vect{x}_i\wedge\frac{1}{\ii}\nabla_{\vect{x}_i}$ (angular momentum operator acting on the position coordinates $\vect{x}_i\in\R^3$ of
the $i$th electron), $L_x,\,L_y,\,L_z$ (components of $\vect{L}$), and analogously for spin (see e.g. Ref.~\onlinecite{NuclearChargeLimit2009}).
The parity operator
$\psi(\vect{x}_1,s_1,\dots,\vect{x}_N,s_N)\mapsto\psi(-\vect{x}_1,s_1,\dots,-\vect{x}_N,s_N)$ is denoted by $\hat{R}$. 

\subsection{Configurations}
\label{sec:FrameworkConfigurations}

Our treatment of symmetry reduction is independent of the particular orbitals used, and works within the context
of general $N$-electron \emph{configurations} as introduced in Ref.~\onlinecite[definition~2.2]{CI2009}: Let
\begin{equation}
V_1,\, V_2, \dots \subset L^2\left(\basespacem\right)
\label{eq:IrredRepLS}
\end{equation}
be any collection of mutually orthogonal subspaces of the single-electron Hilbert space, which are irreducible 
representation spaces for the joint spin and angular momentum algebra 
$\mathrm{Span}\,\{L_x, L_y, L_z, S_x, S_y, S_z\}$ 
(or, equivalently, which are joint eigenspaces of $\vect{L}^2$ and $\vect{S}^2$ with minimal dimension $(2\ell+1)(2s+1)$ 
given the respective eigenvalues $\ell(\ell+1)$ and $s(s+1)$).
Then, a \emph{configuration} of an $N$-electron atom or ion is a subspace of the antisymmetrized $N$-electron state space 
$L_a^2((\basespacet)^N)$ of the following form:
\begin{equation*}
\label{eq:Configs}
\begin{split}
&\mathcal{C}^{d_1,\dots,d_k} = \mathrm{Span}\left\{\ket{\chi_1,\dots,\chi_N} \suchthat \{\chi_1,\dots,\chi_N\}\text{ any ON}\right.\\
&\left.\text{set in }L^2\left(\basespacem\right)\text{ with }\sharp\left\{i \suchthat \chi_i \in V_j\right\} = d_j\right\}
\end{split}
\end{equation*}
where $(d_1,\dots,d_k) = \vect{d}$ is a partition of $N$ (i.e.\ $0\le d_j \le \dim\,V_j,\, \sum_j d_j = N$). Physically, the $V_j$ are
subshells and the $d_j$ are occupation numbers.
The main point is that \emph{all} choices of the $\chi_i$'s consistent with the requirement that a fixed number of them have to be 
picked from each $V_j$ have to be included. 

Configurations, unlike general subspaces of $L_a^2((\basespacet)^N)$ spanned by a basis of 
Slater determinants, are invariant under the symmetry group~\eqref{eq:SymGroup} and its generators $\vect{L}$, $\vect{S}$ and 
$\hat{R}$, and in particular under the operators~\eqref{eq:simLS}. The same holds for multi-configuration subspaces
\begin{equation}
V = \mathrm{Span}\left\{\mathcal{C}^{\vect{d}^{(1)}},\dots,\mathcal{C}^{\vect{d}^{(M)}}\right\}
\label{eq:CIspace}
\end{equation}
where each $C^{\vect{d}^{(i)}}$ is a configuration. 

\subsection{Asymptotics-based selection of configurations}
\label{sec:AsymptoticsBased}

Eq.~\eqref{eq:CIspace} still leaves a great deal of freedom for the precise specification of the CI subspace $V$.
In asymptotics-based CI\cite{CI2009}, the traditional step of an intermediate Hartree-Fock calculation to determine 
orbitals is replaced by the theoretical requirement that the ansatz space reproduce correctly the lowest
Schr\"odinger eigenstates in the iso-electronic limit $Z\to\infty$ (see Theorem~\ref{CItheorem} below). This\\
-- requires Slater type orbitals (STO's) instead of the asymptotically inexact linear combinations of Gaussians which are common in molecular calculations,\\
-- and corresponds to full CI in an active space for the valence electrons (instead of truncating valence electron correlations in terms of order of
excitation with respect to a reference determinant as in double-excited CI, or nonlinearly approximating them as in coupled cluster theory).

Asymptotics-based CI preserves the spin and angular momentum symmetries of the original Hamiltonian, and obeys the virial theorem, 
by determining orbital dilation parameters self-consistently for the actual CI wavefunctions 
instead of precomputing them via a Hartree-Fock calculation. (The fact that methods with self-consistent
dilation parameters always obey the virial theorem was pointed out by L\"owdin\cite{Lowdin1959}.)

The specific asymptotics-based CI model for atoms used in this paper is as follows.
\begin{itemize}
\item[(A)] (Choice of a parametrized, asymptotically exact family of subspaces)
We specify the orbital spaces in Eq.~\eqref{eq:IrredRepLS} as 
\begin{equation*}
\begin{split}
V_{n\ell}^{\vect{Z}} &:= \mathrm{Span}\left\{\psi_{n \ell m}\!\uparrow, \psi_{n \ell m}\!\downarrow\right\}_{m=-\ell\dots\ell},\\
n &= 1,2,\dots, \quad \ell = 0,\dots,n-1
\end{split}
\end{equation*}
with orthonormal Slater (or hydrogen-like) orbitals
\begin{equation*}
\psi_{n \ell m}(\vect{x}) = r^\ell\,Y_{\ell m}(\vartheta,\varphi)\,p_{n\ell}\left(Z_{1\ell},\dots,Z_{n\ell},r\right)\,\mathrm{e}^{-\frac{Z_{n\ell}}{n} r},
\end{equation*}
$r = \abs{\vect{x}}$, and polynomials $p_{n\ell}(Z_{1\ell},\dots,Z_{n\ell},\cdot)$ of order $n-\ell-1$, 
see equation~\eqref{eq:Orbitals}. Here, $\vect{Z} = (Z_{1,0},Z_{2,0},Z_{2,1},\dots)$ is a vector 
of dilation parameters $Z_{n\ell}>0$. We then set 
\begin{equation*}
V^{\vect{Z}} := \mathrm{Span}\left\{\bigcup_{\vect{d}\in\mathcal{D}} \mathcal{C}^{\vect{d}}\right\}
\end{equation*}
where $\vect{d} = (d_{n\ell})$, $n = 1,2,\dots$, $\ell = 0,\dots,n-1$ is a vector of occupation numbers which sum to $N$, and ${\mathcal D}$ is a finite set
of such vectors such that 
\begin{enumerate}
\item[(i)] $0\le d_{n\ell}\le \dim\,V_{n\ell}^{\vect{Z}} = 2\cdot (2\ell + 1)$.
\end{enumerate}
Prototypical is the set ${\mathcal{D}}$ consisting of all configurations such that,
with respect to alphabetical ordering of the indices $(n,\ell)$,
\begin{enumerate}
\item[(ii)] $d_{n\ell}=0$ for $(n,\ell) > (n,\ell)_{\max}$
\item[(iii)] $d_{n\ell}=2\cdot(2\ell+1)$ for $(n,\ell) \le (n,\ell)_{\min}$.
\end{enumerate}
Here~(ii) is a cutoff condition, and~(iii) says that all subshells up to $(n,\ell)_{\min}$ are completely filled.

\item[(B)] (Subspace eigenvalue problem) For each symmetry subspace
\begin{equation*}
\begin{split}
V_{\ell s p}^{\vect{Z}} &:= \left\{\psi \in V^{\vect{Z}} \suchthat \vect{L}^2 \psi = \ell(\ell+1)\psi,\right\}\\
&\left.\vect{S}^2 \psi = s(s+1)\psi,\ \hat{R}\,\psi = p\,\psi\right\}
\end{split}
\end{equation*}
of $V^{\vect{Z}}$ (angular momentum, spin and parity quantum numbers $\ell$, $s$ and $p$, respectively), $E^{\mathrm{CI},\vect{Z}}_j$ := eigenvalues of $P^{\vect{Z}}\,H\,P^{\vect{Z}}$ on $V_{\ell s p}^{\vect{Z}}$, $\psi^{\mathrm{CI},\vect{Z}}_j$ := corresponding orthonormal eigenstates, where $P^{\vect{Z}}$ = orthogonal projector of $L_a^2((\basespacet)^N)$ onto $V^{\vect{Z}}$.
\item[(C)] (Variational parameter determination) For each symmetry subspace $V_{\ell s p}^{\vect{Z}} \subset V^{\vect{Z}}$, $\vect{Z}_\ast := \mathrm{argmin}_\vect{Z}(\min_j E^{\mathrm{CI},\vect{Z}}_j)$, $E^{\mathrm{CI}}_j := E^{\mathrm{CI},\vect{Z}_\ast}_j$, $\psi^{\mathrm{CI}}_j := \psi^{\mathrm{CI},\vect{Z}_\ast}_j$.
\end{itemize}
Here $\mathrm{argmin}_x f(x)$ denotes a minimizer of the functional $f$. We remark that minimizing dilation parameters $\vect{Z}$ are expected to exist provided the nuclear charge $Z$ is greater or equal to the number $N$ of electrons (in which case the full Rayleigh-Ritz variational principle possesses a minimizer~\cite{Zhislin1960}). In our numerical computations we always found this to be the case. 

Also, for future reference we define
\begin{align*}
\begin{split}
c^{\mathrm{CI}} &:= \text{number of core spin-orbitals of the CI model}\\
&= \sum_{(n,\ell) \le (n,\ell)_{\min}} 2(2\ell+1),
\end{split}\\
\begin{split}
t^{\mathrm{CI}} &:= \text{total number of spin-orbitals of the CI model}\\
&= \sum_{(n,\ell) \le (n,\ell)_{\max}} 2\cdot(2\ell+1).
\end{split}
\end{align*}
Of course, the model only makes sense (i.e., the space $V$ is nonempty) provided the cutoffs $(n,\ell)_{\min}$, $(n,\ell)_{\max}$ are chosen so that $c^{\mathrm{CI}} \le N \le t^{\mathrm{CI}}$.

We summarize the asymptotic properties of the above model in the following straightforward generalization
of Theorem~2.1 in Ref.~\onlinecite{CI2009} on second-period atoms.
The following numbers associated with the non-interacting $N$-electron atom play a role: $n_-(N)$, $n_+(N)$, $c(N)$, $t(N)$ which denote the number of closed shells, closed or open shells, core spin-orbitals, and core or valence spin-orbitals, respectively.
Explicitly\cite{NuclearChargeLimit2009}, $n_-$ and $n_+$ can be expressed in terms of the number of spin-orbitals in the first $n'$ hydrogen shells, 
$f(n'):=\sum_{n=1}^{n'}\sum_{\ell=0}^{n-1} 2\cdot(2\ell+1)$, as the largest integer such that $f(n_-) < N$, respectively the smallest integer such that $f(n_+) \ge N$. One then has $c(N) = f(n_-(N))$, $t(N) = f(n_+(N))$.

\begin{theorem}\label{CItheorem} (Correct asymptotic behaviour) 
The CI model (A), (B), (C) with ${\mathcal D}$ given by (ii), (iii) has the following properties.
In the large nuclear charge limit $Z \to \infty$ for $N$ and $\dim V^{\vect{Z}}$ fixed, the lowest
\begin{equation*}
\binom{\min\left(t(N),t^{\mathrm{CI}}\right)-\max\left(c(N),c^{\mathrm{CI}}\right)}{N-\max\left(c(N),c^{\mathrm{CI}}\right)}
\end{equation*}
eigenvalues $E_1^{\mathrm{CI}}\le E_2^{\mathrm{CI}}\le \dots$ and $E_1\le E_2\le \dots$ (repeated according to multiplicity) of the CI model respectively the Schr\"odinger equation~\eqref{eq:Schroedinger} satisfy
\begin{equation*}
\lim\frac{E^{\mathrm{CI}}_j}{E_j} = 1.
\end{equation*}

If moreover $c^{\mathrm{CI}}\le c(N)$ (i.e. the CI model does not constrain the occupation numbers of any non-core orbitals) and $t^{\mathrm{CI}}\ge t(N)$
(i.e. at least all core and valence orbitals are included in the CI model), then there exist orthonormal CI respectively
Schr\"odinger eigenstates $\psi_i^{\mathrm{CI}}$ and $\psi_i$ corresponding to the above eigenvalues such that 
\begin{equation*}
\lim\norm{\psi_i^{\mathrm{CI}}-\psi_i} = 0,
\end{equation*}
where $\norm{\cdot}$ is the norm on the $N$-electron space $L_a^2((\basespacet)^N)$.

Finally, under the same condition on $c^{\mathrm{CI}}$ and $t^{\mathrm{CI}}$ the spectral gaps satisfy
\begin{equation} \label{eq:gaps}
\lim\frac{\Delta E^{\mathrm{CI}}_j}{\Delta E_j} = 1
\end{equation}
whenever $\Delta E_j > 0$, where $\Delta E^{\mathrm{CI}}_j = E^{\mathrm{CI}}_j - E^{\mathrm{CI}}_1$ and $\Delta E_j = E_j - E_1$ ($j \ge 2$).
\end{theorem}

We emphasize that the above theorem only covers the regime of large $Z$. For neutral atoms, the highest eigenstates in the ansatz space of asymptotics-based CI are typically observed to lie above higher Rydberg states or even above the bottom of the continuous spectrum.

\section{Exact reduction steps and LS diagonalization}
\label{sec:Reduction}

This section explains exact reduction steps which are essential for cutting down the calculation
time and storage requirement of the algorithmic implementation.

\subsection{Tensor product structure of configurations}
\label{sec:ReductionConfigurations}

Our first observation connects $N$-particle configurations (see Section~\ref{sec:FrameworkConfigurations}) to
the \emph{non-antisymmetrized} tensor product of antisymmetrized $d_j$-particle states, preserving the action 
of the angular momentum and spin operators. Here and below, we use the standard notation\cite{GlimmJaffe} 
$\wedge^n V$ for the $n$-fold antisymmetrized tensor product of a vector space $V$, and $V \otimes W$
for the tensor product of two spaces $V$ and $W$.

\begin{proposition}
Consider irreducible representation spaces $V_1, \dots, V_k$ as in equation~\eqref{eq:IrredRepLS}
and particle numbers $d_1,\dots,d_k \ge 0$. Then the following isometric isomorphism holds,
\begin{equation}
\mathcal{C}^{d_1,\dots,d_k} \cong \bigotimes_{j=1}^k \wedge^{d_j} V_j.
\label{eq:ConfigProdSubshells}
\end{equation}
A canonical mapping of basis vectors is given by
\begin{equation*}
\begin{split}
\mathcal{T}:
&\ket{\chi^1_1,\dots,\chi^1_{d_1},\dots,\chi^k_{d_k}} \mapsto\\
&\ket{\chi^1_1,\dots,\chi^1_{d_1}} \otimes \cdots \otimes \ket{\chi^k_1,\dots,\chi^k_{d_k}}
\end{split}
\end{equation*}
with $\chi^j_i \in V_j$ for all $i,j$. Moreover, $\mathcal{T}$ commutes with the action
of the angular momentum and spin operators, i.e.,
\begin{equation*}
\mathcal{T}\left(\vect{L}\,\psi\right) = \left(\sum_{j=1}^k \vect{L}_j\right) \mathcal{T}(\psi),\quad
\mathcal{T}\left(\vect{S}\,\psi\right) = \left(\sum_{j=1}^k \vect{S}_j\right) \mathcal{T}(\psi)
\end{equation*}
for all $\psi \in \mathcal{C}^{d_1,\dots,d_k}$ where $\vect{L}, \vect{S}$ on the left hand sides are $N$-particle operators and each
$\vect{L}_j, \vect{S}_j$ on the right hand side acts on the $d_j$-particle
tensor factor $\wedge^{d_j} V_j$.
\label{prop:IsoConfigTensor}
\end{proposition}
\begin{proof}
Clear from the definitions.
\end{proof}
In particular, $\dim(\mathcal{C}^{d_1,\dots,d_k}) = \prod_j \binom{\dim(V_j)}{d_j}$. 
Note that equation~\eqref{eq:ConfigProdSubshells} inherently takes into account the antisymmetrization of 
fermionic wavefunctions, without requiring any additional normalization factors. 
The isometry~\eqref{eq:ConfigProdSubshells} is reflected in the algorithmic implementation by ordering Slater 
determinants lexicographically and arranging coefficients accordingly, see
Figure~\ref{fig:lex_slater_ordering}.

\begin{figure}[ht!]
\begin{center}
\includegraphics[width=0.4\textwidth]{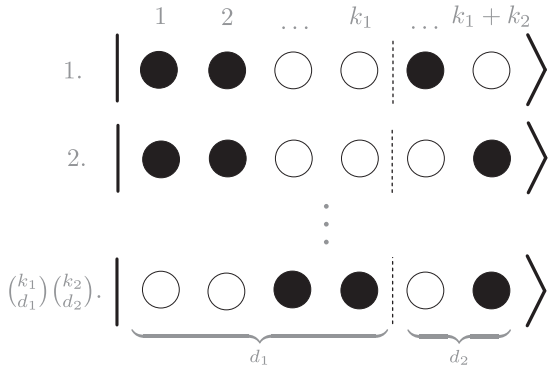}
\end{center}
\caption{Lexicographical ordering of the $\binom{k_1}{d_1}\cdot\binom{k_2}{d_2}$ 
Slater determinants restricted to a fixed configuration involving 
two one-particle subspaces $V_1$, $V_2$ of dimension $k_1$, $k_2$, with $d_1$ particles in 
orbitals $1,\dots,k_1$ and $d_2$ particles in the remaining orbitals $k_1+1,\dots,k_1+k_2$. Filled circles correspond to filled orbitals.
Usefully, restriction to a configuration preserves lexicographical ordering.}
\label{fig:lex_slater_ordering}
\end{figure}

\subsection{LS diagonalization}
\label{sec:LSDiag}

A priori, the diagonalization of the angular momentum and spin (LS) operators seems as expensive
as diagonalizing the Hamiltonian itself, yet it turns out to come at much cheaper costs. It
involves mostly algebra, and can be done prior to setting up the Hamiltonian.
\smallskip

1. Calculate all irreducible LS-eigenspaces for each many-particle subshell.
In more detail, let $u = s,p,d,\dots$ denote the angular momentum subshells in
common chemist's notation and set
\begin{equation*}
V_u := \mathrm{Span}\left\{Y_{\ell m}\!\uparrow,Y_{\ell m}\!\downarrow\right\}_{\ell=\mathrm{ang}(u),m=-\ell \dots \ell},
\end{equation*}
with the spherical harmonics $Y_{\ell m}$. Note that this is an explicit realization of the spaces in Eq.~\eqref{eq:IrredRepLS}. Then, for all $n = 1,\dots,\dim(V_u)$ (equal to $2 \times (2\,\ell + 1)$), decompose the $n$-particle space $\wedge^n V_u$ into the direct sum of irreducible spin and angular momentum representation spaces. That is,
\begin{equation}
\label{eq:IrredSubshell}
\wedge^n V_u = \bigoplus_i V_{uni}
\end{equation}
such that
\begin{equation*}
\begin{split}
\vect{L}^2\,\varphi &= \ell_i \left(\ell_i+1\right)\,\varphi,\\
\vect{S}^2\,\varphi &= s_i \left(s_i+1\right)\,\varphi \quad \forall\,\varphi \in V_{uni},\\
\dim(V_{uni}) &= \left(2\ell_i+1\right)\left(2s_i+1\right).
\end{split}
\end{equation*}
Explicit results are shown in Table~\ref{tab:irredLS}: subshells from $\wedge^4 V_d$ to $\wedge^{10} V_d$ are omitted for brevity's sake, and only states with maximal $L_z$ and $S_z$ quantum numbers are displayed; applying the ladder operators $L_- = L_x - i L_y$ and $S_- = S_x - i S_y$ yields the remaining wavefunctions. Note that symmetry levels can appear twice within a many-particle subshell, e.g., $^2D$ in $\wedge^3 V_d$. In concordance with the Clebsch-Gordan method below, the ordered single-particle orbitals are $L_z$ and $S_z$ eigenstates, denoted by
\begin{align*}
&\left(s, \conj{s}\right) &\text{ for } V_s,\\
&\left(p_1, \conj{p_1}, p_0, \conj{p_0}, p_{n\!1}, \conj{p_{n\!1}}\right) &\text{ for } V_p,\\
&\left(d_2, \conj{d_2}, d_1, \conj{d_1},\dots, d_{n\!2}, \conj{d_{n\!2}}\right) &\text{ for } V_d.
\end{align*}
The highest quantum number appears first, and $\conj{\,\cdot\,}$ equals spin down $\downarrow$ (convention as in Ref.~\onlinecite{NuclearChargeLimit2009}).

The decomposition~\eqref{eq:IrredSubshell} first requires a matrix representation of the angular momentum and spin operators $L_x, L_y, L_z, S_x, S_y, S_z$ on 
$\wedge^n V_u$. Obtain it by starting from the canonical single-particle representation on $V_u$ (spherical harmonics) and writing the $n$-body operator in the form 
$B = \sum_{i,j} b_{ij}\,a^\dagger_i a_j$, where $b_{ij}$ are the coefficients of the single-particle representation and $a^\dagger_i$, $a_j$ are fermionic creation
and annihilation operators. The operators $a^\dagger_i a_j$ map Slater determinants to Slater determinants; thus all entries of their corresponding matrix representation are $0$ or $\pm 1$.

The next task to arrive at~\eqref{eq:IrredSubshell} involves the simultaneous diagonalization of the pairwise commuting operators $\vect{L}^2, \vect{S}^2, L_z, S_z$. We present two alternatives.

\emph{Alternative 1}
\begin{algorithmic}
\STATE apply an algorithm of choice, e.g. Ref.~\onlinecite{SimDiag1993}, for the simultaneous diagonalization of $\vect{L}^2, \vect{S}^2, L_z, S_z$ on $\wedge^n V_u$, denoting the eigenvalues or "quantum numbers" of $\vect{L}^2, \vect{S}^2, L_z, S_z$ by $\ell(\ell+1), s(s+1), m_\ell, m_s$, respectively
\FOR{each subspaces $W$ with $m_\ell = \ell$ and $m_s = s$}
	\STATE choose an ONB $\left\{\varphi_1,\dots,\varphi_r\right\}$ of $W$
	\FOR{$j=1,\dots,r$}
		\STATE add $V_{unj} := \mathrm{Span}\left\{\varphi_j,L_- \varphi_j, S_- \varphi_j, L_- S_- \varphi_j, \dots\right\}$ to the decomposition~\eqref{eq:IrredSubshell}
	\ENDFOR
\ENDFOR
\label{alg:SimDiagDirect}
\end{algorithmic}
Note that the iterative application of the ladder operators $L_-$ and $S_-$ ensures that the resulting subspaces $V_{unj}$ are invariant irreducible representation spaces.

\emph{Alternative 2}
\begin{algorithmic}
\STATE $\mathrm{count} \gets 0$
\FOR{$\ell = 0,1,\dots$ and $s = \left\{\begin{array}{ll}\frac{1}{2},\frac{3}{2},\dots& n \text{ odd}\\0,1,\dots& n \text{ even}\end{array}\right.$}
	\STATE \begin{equation*}\begin{split}X := &\left(\vect{L}^2 - \ell(\ell+1)\,\mathrm{id}\,\vert\, \vect{S}^2-s(s+1)\,\mathrm{id}\,\vert\right.\\ &\quad\left.\vphantom{\vect{L}^2}L_z - \ell\,\mathrm{id} \,\vert\, S_z - s\,\mathrm{id}\right)\end{split}\end{equation*}
	\STATE calculate $W := \mathrm{Kern}\left(X X^\dagger\right) \subset \wedge^n V_u$
	\IF{$W = \emptyset$}
		\STATE next for loop
	\ENDIF
	\STATE choose an ONB $\left\{\varphi_1,\dots,\varphi_r\right\}$ of $W$\\
	\FOR{$j=1,\dots,r$}
		\STATE add $V_{unj} := \mathrm{Span}\left\{\varphi_j,L_- \varphi_j, S_- \varphi_j, L_- S_- \varphi_j, \dots\right\}$ to the decomposition~\eqref{eq:IrredSubshell}
	\ENDFOR
	\STATE $\mathrm{count} \gets \mathrm{count} + r$
	\STATE stop if $\sum_{i=1}^{\mathrm{count}} \dim\left(V_{uni}\right) = \dim\left(\wedge^n V_u\right)$
\ENDFOR
\label{alg:SimDiagKernel}
\end{algorithmic}
The second alternative exchanges the direct diagonalization in alternative 1 for testing all potential eigenvalues (that is, integer or half-integer numbers) with $m_\ell = \ell$ and $m_s = s$. Efficient numerical methods exist for computing the kernel $W$, which take advantage of the sparse structure of the matrix representation.

\begin{table}[ht!]
\centering
\begin{tabular}{|c|c|cc|c|}
\hline
Config&Sym&$L_z$&$S_z$&$\Psi$ $\vphantom{\bigl(}$\\
\hline\hline
$\wedge^{1}V_s$&$^2S$&$0$&$\frac{1}{2}$&$\ket{s}$ $\vphantom{\bigl(}$\\
\hline
$\wedge^{2}V_s$&$^1S$&$0$&$0$&$\ket{s \conj{s}}$ $\vphantom{\bigl(}$\\
\hline
\hline
$\wedge^{1}V_p$&$^2P^o$&$1$&$\frac{1}{2}$&$\ket{p_1}$ $\vphantom{\bigl(}$\\
\hline
$\wedge^{2}V_p$&$^1S$&$0$&$0$&$\frac{1}{\sqrt{3}}\left(-\ket{p_1 \conj{p_{n\!1}}}+\ket{\conj{p_1} p_{n\!1}}+\ket{p_0 \conj{p_0}}\right)$ $\vphantom{\bigl(}$\\
\cline{2-5}
&$^3P$&$1$&$1$&$\ket{p_1 p_0}$ $\vphantom{\bigl(}$\\
\cline{2-5}
&$^1D$&$2$&$0$&$\ket{p_1 \conj{p_1}}$ $\vphantom{\bigl(}$\\
\hline
$\wedge^{3}V_p$&$^4S^o$&$0$&$\frac{3}{2}$&$\ket{p_1 p_0 p_{n\!1}}$ $\vphantom{\bigl(}$\\
\cline{2-5}
&$^2P^o$&$1$&$\frac{1}{2}$&$\frac{1}{\sqrt{2}}\left(\ket{p_1 \conj{p_1} p_{n\!1}}+\ket{p_1 p_0 \conj{p_0}}\right)$ $\vphantom{\bigl(}$\\
\cline{2-5}
&$^2D^o$&$2$&$\frac{1}{2}$&$\ket{p_1 \conj{p_1} p_0}$ $\vphantom{\bigl(}$\\
\hline
$\wedge^{4}V_p$&$^1S$&$0$&$0$&$\frac{1}{\sqrt{3}}\left(-\ket{p_1 \conj{p_1} p_{n\!1} \conj{p_{n\!1}}}-\ket{p_1 p_0 \conj{p_0} \conj{p_{n\!1}}}\right.$ $\vphantom{\bigl(}$\\
&&&&$\left.+\ket{\conj{p_1} p_0 \conj{p_0} p_{n\!1}}\right)$ $\vphantom{\bigl(}$\\
\cline{2-5}
&$^3P$&$1$&$1$&$\ket{p_1 \conj{p_1} p_0 p_{n\!1}}$ $\vphantom{\bigl(}$\\
\cline{2-5}
&$^1D$&$2$&$0$&$\ket{p_1 \conj{p_1} p_0 \conj{p_0}}$ $\vphantom{\bigl(}$\\
\hline
$\wedge^{5}V_p$&$^2P^o$&$1$&$\frac{1}{2}$&$\ket{p_1 \conj{p_1} p_0 \conj{p_0} p_{n\!1}}$ $\vphantom{\bigl(}$\\
\hline
$\wedge^{6}V_p$&$^1S$&$0$&$0$&$\ket{p_1 \conj{p_1} p_0 \conj{p_0} p_{n\!1} \conj{p_{n\!1}}}$ $\vphantom{\bigl(}$\\
\hline
\hline
$\wedge^{1}V_d$&$^2D$&$2$&$\frac{1}{2}$&$\ket{d_2}$ $\vphantom{\bigl(}$\\
\hline
$\wedge^{2}V_d$&$^1S$&$0$&$0$&$\frac{1}{\sqrt{5}}\left(\ket{d_2 \conj{d_{n\!2}}}-\ket{\conj{d_2} d_{n\!2}}\right.$ $\vphantom{\bigl(}$\\
&&&&$\left.-\ket{d_1 \conj{d_{n\!1}}}+\ket{\conj{d_1} d_{n\!1}}\right.$ $\vphantom{\bigl(}$\\
&&&&$\left.+\ket{d_0 \conj{d_0}}\right)$ $\vphantom{\bigl(}$\\
\cline{2-5}
&$^3P$&$1$&$1$&$\frac{1}{\sqrt{5}}\left(-\sqrt{2}\cdot\ket{d_2 d_{n\!1}}+\sqrt{3}\cdot\ket{d_1 d_0}\right)$ $\vphantom{\bigl(}$\\
\cline{2-5}
&$^1D$&$2$&$0$&$\frac{1}{\sqrt{7}}\left(-\sqrt{2}\cdot\ket{d_2 \conj{d_0}}+\sqrt{2}\cdot\ket{\conj{d_2} d_0}\right.$ $\vphantom{\bigl(}$\\
&&&&$\left.+\sqrt{3}\cdot\ket{d_1 \conj{d_1}}\right)$ $\vphantom{\bigl(}$\\
\cline{2-5}
&$^3F$&$3$&$1$&$\ket{d_2 d_1}$ $\vphantom{\bigl(}$\\
\cline{2-5}
&$^1G$&$4$&$0$&$\ket{d_2 \conj{d_2}}$ $\vphantom{\bigl(}$\\
\hline
$\wedge^{3}V_d$&$^2P$&$1$&$\frac{1}{2}$&$\frac{1}{\sqrt{210}}\left(4 \sqrt{3}\cdot\ket{d_2 d_1 \conj{d_{n\!2}}}-2 \sqrt{3}\cdot\ket{d_2 \conj{d_1} d_{n\!2}}\right.$ $\vphantom{\bigl(}$\\
&&&&$\left.-4 \sqrt{2}\cdot\ket{d_2 d_0 \conj{d_{n\!1}}}-\sqrt{2}\cdot\ket{d_2 \conj{d_0} d_{n\!1}}\right.$ $\vphantom{\bigl(}$\\
&&&&$\left.-2 \sqrt{3}\cdot\ket{\conj{d_2} d_1 d_{n\!2}}+5 \sqrt{2}\cdot\ket{\conj{d_2} d_0 d_{n\!1}}\right.$ $\vphantom{\bigl(}$\\
&&&&$\left.+3 \sqrt{3}\cdot\ket{d_1 \conj{d_1} d_{n\!1}}+3 \sqrt{3}\cdot\ket{d_1 d_0 \conj{d_0}}\right)$ $\vphantom{\bigl(}$\\
\cline{2-5}
&$^4P$&$1$&$\frac{3}{2}$&$\frac{1}{\sqrt{5}}\left(-\sqrt{3}\cdot\ket{d_2 d_1 d_{n\!2}}+\sqrt{2}\cdot\ket{d_2 d_0 d_{n\!1}}\right)$ $\vphantom{\bigl(}$\\
\cline{2-5}
&$^2D$&$2$&$\frac{1}{2}$&$\frac{1}{\sqrt{15}}\left(2 \sqrt{2}\cdot\ket{d_2 \conj{d_2} d_{n\!2}}-\sqrt{2}\cdot\ket{d_2 \conj{d_1} d_{n\!1}}\right.$ $\vphantom{\bigl(}$\\
&&&&$\left.+\sqrt{2}\cdot\ket{\conj{d_2} d_1 d_{n\!1}}+\sqrt{3}\cdot\ket{d_1 \conj{d_1} d_0}\right)$ $\vphantom{\bigl(}$\\
\cline{2-5}
&$^2D$&$2$&$\frac{1}{2}$&$\frac{1}{\sqrt{70}}\left(-\ket{d_2 \conj{d_2} d_{n\!2}}-5\cdot\ket{d_2 d_1 \conj{d_{n\!1}}}\right.$ $\vphantom{\bigl(}$\\
&&&&$\left.+3\cdot\ket{d_2 \conj{d_1} d_{n\!1}}+5\cdot\ket{d_2 d_0 \conj{d_0}}\right.$ $\vphantom{\bigl(}$\\
&&&&$\left.+2\cdot\ket{\conj{d_2} d_1 d_{n\!1}}+\sqrt{6}\cdot\ket{d_1 \conj{d_1} d_0}\right)$ $\vphantom{\bigl(}$\\
\cline{2-5}
&$^2F$&$3$&$\frac{1}{2}$&$\frac{1}{2 \sqrt{3}}\left(\sqrt{6}\cdot\ket{d_2 \conj{d_2} d_{n\!1}}-\ket{d_2 d_1 \conj{d_0}}\right.$ $\vphantom{\bigl(}$\\
&&&&$\left.-\ket{d_2 \conj{d_1} d_0}+2\cdot\ket{\conj{d_2} d_1 d_0}\right)$ $\vphantom{\bigl(}$\\
\cline{2-5}
&$^4F$&$3$&$\frac{3}{2}$&$\ket{d_2 d_1 d_0}$ $\vphantom{\bigl(}$\\
\cline{2-5}
&$^2G$&$4$&$\frac{1}{2}$&$\frac{1}{\sqrt{5}}\left(\sqrt{2}\cdot\ket{d_2 \conj{d_2} d_0}+\sqrt{3}\cdot\ket{d_2 d_1 \conj{d_1}}\right)$ $\vphantom{\bigl(}$\\
\cline{2-5}
&$^2H$&$5$&$\frac{1}{2}$&$\ket{d_2 \conj{d_2} d_1}$ $\vphantom{\bigl(}$\\
\hline
\end{tabular}
\caption{Irreducible LS eigenspace decompositions of $\wedge^n V_u$ in Eq.~\eqref{eq:IrredSubshell}, showing states with maximal $L_z$ and $S_z$ quantum numbers only.}
\label{tab:irredLS}
\end{table}

\smallskip
2. Consider $N$-electron configurations $\mathcal{C}^{n_1,\dots,n_k}$ assembled from the above single-particle subshells, with $n_j$ electrons in subshell $j$ (angular momentum $u_j$) such that $N = \sum_j n_j$. Using the decomposition in step 1, simultaneously diagonalize the pairwise commuting operators~\eqref{eq:simLS} acting on $\mathcal{C}^{n_1,\dots,n_k}$. (The parity operator $\hat{R}$ is constant on $\mathcal{C}^{n_1,\dots,n_k}$ anyway and needs no further consideration.) In more detail, the isometry~\eqref{eq:ConfigProdSubshells} and the decomposition~\eqref{eq:IrredSubshell} imply
\begin{equation}
\mathcal{C}^{n_1,\dots,n_k} \simeq \bigoplus_{I = \left(i_1,\dots,i_k\right)} V_I, \quad V_I := \bigotimes_j V_{u_j,n_j,i_j}.
\label{eq:TensorProdConfig}
\end{equation}
By construction, each $V_I$ is uniquely characterized by its eigenvalues with respect to the LS-operators $\vect{L}_j^2$ and $\vect{S}_j^2$ acting on the $j$th tensor factor. Since
\begin{equation*}
\vect{L} = \sum_j \vect{L}_j,\quad \vect{S} = \sum_j \vect{S}_j,
\end{equation*}
all operators
\begin{equation*}
\vect{L}^2,\, \vect{S}^2,\, L_z,\, S_z,\, \hat{R},\, \vect{L}_j^2,\, \vect{S}_j^2\, \quad j=1,\dots,k
\end{equation*}
commute pairwise, and it follows that each $V_I$ is an invariant subspace of the operators~\eqref{eq:simLS}. Thus, the diagonalization can be performed on each $V_I$ independently.

An explicit solution for the diagonalization in case of $k = 2$ is well known in terms of the Clebsch-Gordan coefficients, which can be iteratively extended to higher $k$. We obtain
\begin{equation}
V_I = \bigoplus_{\substack{\ell\,s\,m_\ell\,m_s\\ \abs{m_\ell}\le \ell,\abs{m_s}\le s}} V_{I,\ell\,s\,m_\ell\,m_s}
\label{eq:SimDiagVI}
\end{equation}
such that for all $\varphi \in V_{I,\ell\,s\,m_\ell\,m_s}$,
\begin{align*}
\vect{L}^2\,\varphi &= \ell(\ell + 1)\,\varphi, &L_z\,\varphi &= m_\ell\,\varphi\\
\vect{S}^2\,\varphi &= s(s + 1)\,\varphi, &S_z\,\varphi &= m_s\,\varphi.
\end{align*}
Note that $V_{I,\ell\,s\,m_\ell\,m_s}$ may be zero for some $\ell,s,m_\ell,m_s$.

Assembling equations~\eqref{eq:TensorProdConfig} and~\eqref{eq:SimDiagVI}, we obtain
\begin{equation*}
\begin{split}
\mathcal{C}^{n_1,\dots,n_k} &\simeq \bigoplus_{\ell\,s\,m_\ell\,m_s} V_{\ell\,s\,m_\ell\,m_s},\\
V_{\ell\,s\,m_\ell\,m_s} &:= \bigoplus_I V_{I,\ell\,s\,m_\ell\,m_s}.
\end{split}
\end{equation*}
That is, we have decomposed the configurations into the simultaneous eigenspaces of the angular momentum and spin 
operators~\eqref{eq:simLS}.

\subsection{Restriction to fixed $m_\ell$ and $m_s$} 
\label{sec:RepresentationTheory}

From general results about the angular momentum and spin algebra, it is well known that within an irreducible $\vect{L}^2$-$\vect{S}^2$-eigenspace, the ladder operators $L_\pm = L_x \pm i L_y$ and $S_\pm = S_x \pm i S_y$ traverse the $L_z$ and $S_z$ eigenstates, respectively. Additionally, the ladder operators commute with the Hamiltonian $H$ in~\eqref{eq:Hamiltonian} as well as with the CI Hamiltonian. Thus, in terms of eigenvalue determination, it suffices to restrict to LS eigenstates with fixed $m_\ell$ and $m_s$. We adopt the convention in Ref.~\onlinecite{NuclearChargeLimit2009}, and set $m_\ell \equiv 0, m_s \equiv s$ in the sequel.

\subsection{Reduced density matrices (RDMs)}
\label{sec:RDMs}

In this subsection, we will incorporate RDMs (see e.g. Ref.~\onlinecite{Lowdin1955,Mazziotti2007} and~\onlinecite{Ando1963,*Coleman1963}) into the algorithmic 
framework to gain computational speedups and memory storage savings. In fact, we use RDM's of wavefunction \emph{pairs}. 


For any pair of states $\psi$ and $\chi$ in the $N$-body Hilbert space~\eqref{eq:Schroedinger}, the matrix element of the Hamiltonian~\eqref{eq:Hamiltonian} can be rewritten as
\begin{equation}
\hprod{\chi}{H\psi} = \trace_{\mathcal{H}}\left[h_0\,\gamma_{\ket{\psi}\bra{\chi}}\right] + \trace_{\wedge^2 \mathcal{H}}\left[v_{ee}\,\Gamma_{\ket{\psi}\bra{\chi}}\right],
\label{eq:AvrHrdm}
\end{equation}
where $\gamma_{\ket{\psi}\bra{\chi}}$ and $\Gamma_{\ket{\psi}\bra{\chi}}$ are the one- and two-body reduced density matrices of the $N$-body matrix $\ket{\psi}\bra{\chi}$, respectively. Here $h_0$ is the single-particle (hydrogen-like) Hamiltonian and $v_{ee}$ is the interelectronic Coulomb potential,
\begin{equation}
h_0 = -\frac{1}{2}\,\Delta_{\vect{x}} - \frac{Z}{\abs{\vect{x}}}, \quad v_{ee} = \frac{1}{\abs{\vect{x}-\vect{y}}}.
\label{eq:h0_vee}
\end{equation}
Since these operators are independent of spin, we may effectively "trace out" the spin. With the standard notation
\begin{equation}
\bprod{a b}{c d} := \int_{\R^6} \conj{a(\vect{x}_1)} b(\vect{x}_1)\,\frac{1}{\abs{\vect{x}_1-\vect{x}_2}}\,\conj{c(\vect{x}_2)} d(\vect{x}_2) \ud \vect{x}_1 \vect{x}_2,
\label{eq:CoulombIntDef}
\end{equation}
we obtain
\begin{equation}
\label{eq:EvalSingle}
\hprod{\chi}{H\,\vert\,\psi} = \trace\left[\hat{h}_0 \hat{\gamma}_{\ket{\psi}\bra{\chi}}\right]
+ \trace\left[\hat{v}_{ee}\,\hat{\Gamma}_{\ket{\psi}\bra{\chi}}\right], 
\end{equation}
with
\begin{align}
\label{eq:h0Spintrace}
\left(\hat{h}_0\right)_{i,j} &:= \hprod{i}{h_0 \,\vert\, j},\\
\label{eq:Gamma1Spintrace}
\left(\hat{\gamma}_{\ket{\psi}\bra{\chi}}\right)_{i,j} &:= \sum_{\alpha}
\hprod{i\alpha}{\gamma_{\ket{\psi}\bra{\chi}} \,\vert\, j\alpha},\\
\label{eq:veeSpintrace}
\left(\hat{v}_{ee}\right)_{ij,k\ell} &:= \bprod{i j}{k \ell},\\
\label{eq:Gamma2Spintrace}
\left(\hat{\Gamma}_{\ket{\psi}\bra{\chi}}\right)_{k\ell,ij} &:= \sum_{\substack{\alpha,\beta\\i\alpha < k\beta}} \hprod{j\alpha,\ell\beta}{\Gamma_{\ket{\psi}\bra{\chi}}\,\vert\,i\alpha,k\beta}.
\end{align}
Here $i,j,k,\ell$ denote spatial orbitals and $\alpha,\beta,\gamma,\delta$ are associated spin-parts. The inequality constraint in the last sum refers to lexicographical ordering of spin-orbitals.

By choosing the spatial orbitals real-valued, it follows that $\bprod{i j}{k \ell} = \bprod{j i}{k \ell}$ and $\bprod{i j}{k \ell} = \bprod{i j}{\ell k}$ for all $i,j,k,\ell$. Thus, together with $\bprod{i j}{k \ell} \equiv \bprod{k \ell}{i j}$, it suffices to calculate $\bprod{i j}{k \ell}$ for $i \le j$, $k \le \ell$ and $(i,j) \le (k,\ell)$ (in lexicographical order) only.
\smallskip

For our purposes, the following two features of the above RDM formalism are crucial. First, it avoids having to set up the full $N$-particle
operators $H_0$ and $V_{ee}$, allowing one to work instead with the one- and two-particle operators $h_0$ and $v_{ee}$; this leads to significant
storage savings, see Section~\ref{sec:StorageSaving}. Second, the map from $\psi$ and $\chi$ to
$\hat{\Gamma}_{\ket{\psi}\bra{\chi}}$ is an algebraic coefficient mapping which only depends on the symmetry types of the orbitals (i.e., 
$s$, $p$, $d$, ...) and neither the radial wavefunctions nor the dilation parameters $Z_{n\ell}$. So 
the $\hat{\Gamma}_{\ket{\psi}\bra{\chi}}$ can be precomputed 
for each angular momentum and spin symmetry eigenspace, without any reference to the Hamiltonian. The dilation parameters 
only enter the stage via the Couloumb integrals in $\hat{v}_{ee}$.

\section{Handling Slater orbitals (STO's)}
\label{sec:Orbitals}

\subsection{Orthonormalization}
\label{sec:Orthonormalization}

In this subsection we formalize the orthonormalization calculations for Slater-type orbitals~(STOs) 
employed in Ref.~\onlinecite[equation~(31)]{NuclearChargeLimit2009}. There, only $1s$,$2s$ and $2p$ wavefunctions are considered, whereas here, we handle arbitrary subshells.

More concretely, the wave functions are given by
\begin{equation}
\begin{split}
\psi_{n \ell m}(\vect{x}) &= s_{n\ell}\,r^\ell\,Y_{\ell m}(\vartheta,\varphi) \left(\sum_{i=0}^{n-\ell-1} b_{n\ell,i}\,c_{n\ell,i}\,r^i\right) \mathrm{e}^{-\frac{Z_{n\ell}}{n} r},\\
&r = \abs{\vect x}, \quad \ell = 0,\dots,n-1, \quad n = 1, 2, \dots
\label{eq:Orbitals}
\end{split}
\end{equation}
with $b_{n\ell,i}$ being the $i$th coefficient of the associated Laguerre polynomial $p(r) = L_{n-\ell-1}^{2\ell+1}\left(\frac{2 r}{n}\right)$,
\begin{equation*}
b_{n\ell,i} := \binom{n+\ell}{2\ell+1+i} \frac{(-2/n)^i}{i!}
\end{equation*}
and to-be determined orthogonalization coefficients $c_{n\ell,i} \in \R$ ($i=0,\dots,n-\ell-1$) as well as orthonormalization constants $s_{n\ell} > 0$. Since the spherical harmomics $Y_{\ell m}$ are orthogonal, we may fix the angular momentum quantum numbers $\ell, m$. Now using $\int_0^\infty r^n \mathrm{e}^{-\lambda r} \ud r = \frac{n!}{\lambda^{n+1}}$, orthogonality translates to
\begin{equation}
\begin{split}
0 &= \hprod{\psi_{n \ell m}}{\psi_{k \ell m}}\\
&= \int_0^\infty \int_0^{2\pi} \int_0^{\pi} \conj{\psi_{n \ell m}(\vect{x})}\,\psi_{k \ell m}(\vect{x})\,r^2 \sin\vartheta \ud\vartheta \ud\varphi \ud r\\
&= \conj{s_{n\ell}} s_{k\ell} \sum_{i,j}\,b_{n\ell,i}b_{k\ell,j}\,\conj{c_{n\ell,i}}\,c_{k\ell,j}\, \frac{(i+j+2\ell+2)!}{\left(\frac{Z_{n\ell}}{n} + \frac{Z_{k\ell}}{k}\right)^{i+j+2\ell+3}}\\
&= \conj{s_{n\ell}} s_{k\ell} \hprod{c_{n\ell}}{B_{n\ell} H_{nk}^\ell B_{k\ell}\,c_{k\ell}}
\end{split}
\label{eq:OrbsPerp}
\end{equation}
for all $k = \ell+1,\dots,n-1$. Here we have extended the vectors $(c_{k\ell,i})_{i=0,\dots,k-\ell-1}$ by $c_{k\ell,i} = 0$ for $i \ge k-\ell$. The Hankel matrix $H_{nk}^\ell$ is defined by
\begin{equation*}
H_{nk}^\ell := \left(a_{i+j}^\ell(\lambda)\right)_{i,j}\evalat_{\lambda = \frac{Z_{n\ell}}{n} + \frac{Z_{k\ell}}{k}}, \quad a_i^\ell(\lambda) := \frac{(i+2(\ell+1))!}{\lambda^{i+2\ell+3}}
\end{equation*}
and $B_{n\ell}$ is the diagonal matrix $\mathrm{diag}\left(b_{n\ell,i}\right)_i$. Summarizing Eq.~\eqref{eq:OrbsPerp}, we obtain
\begin{equation}
c_{n\ell} \perp \mathrm{Span}\left(B_{n\ell} H_{nk}^\ell B_{k\ell}\,c_{k\ell}\right)_{k=\ell+1,\dots,n-1},
\label{eq:IterOrthogonal}
\end{equation}
so $c_{n\ell}$ can be calculated iteratively for $n=\ell+1,\ell+2,\dots$, starting from the convention $c_{(\ell+1)\ell,0} = 1$.

Note that the $a_i^\ell$ are the moments of a nonnegative measure $m$ on the positive real axis $\R_+$. Namely, let $\ud m_{\lambda,\ell}(t) = t^{2(\ell+1)}\,\mathrm{e}^{-\lambda t} \ud t$, then
\begin{equation*}
a_i^\ell(\lambda) = \int_{\R_+} t^i \ud m_{\lambda,\ell}(t).
\end{equation*}
The Stieltjes moment problem\cite{Lax2002} states that this is equivalent to the quadratic form given by $H_{nk}^\ell$ being positive.

Once all $c_{n\ell}$ have been obtained, we may plug $k = n$ into~\eqref{eq:OrbsPerp} to calculate the normalization factors $s_{n\ell}$ from
\begin{equation}
1 \stackrel{!}{=} \hprod{\psi_{n \ell m}}{\psi_{n \ell m}} = \abs{s_{n\ell}}^2 \hprod{c_{n\ell}}{B_{n\ell} H_{nn}^\ell B_{n\ell}\,c_{n\ell}}.
\label{eq:OrbsNorm}
\end{equation}

\subsection{One-body integrals}
\label{sec:One-bodyintegrals}
The one-body matrix elements
\begin{equation}
\label{eq:OneBodyInts}
\begin{split}
&\bprod{\psi_{n\ell m}}{\psi_{n'\ell'm'}} := \hprod{\psi_{n\ell m}}{h_0\,\vert\,\psi_{n'\ell'm'}}\\
&= \int_{\R^3} \left(\frac{1}{2}\conj{\nabla\psi_{n\ell m}}\cdot\nabla\psi_{n'\ell'm'} - \frac{Z}{\abs{\vect x}}\,\conj{\psi_{n\ell m}}\,\psi_{n'\ell'm'}\right) \ud^3\vect{x}
\end{split}
\end{equation}
can be evaluated symbolically from Def.~\eqref{eq:Orbitals} by a computer algebra system, via symbolic differentiation and
exact integration in spherical polar coordinates. 

\subsection{Two-body integrals}
\label{sec:Two-bodyintegrals}

\subsubsection{Switching to real-valued, cartesian coordinates}
\label{sec:ToRealCartesian}

As mentioned in Section~\ref{sec:RDMs}, we save computational costs by switching to real-valued spatial orbitals when calculating Coulomb integrals. Thus, for each fixed $\ell$, we apply a unitary base change $U_\ell = (u_{\ell,mm'})_{mm'}$ to the spherical harmonics $Y_{\ell m}$ of degree $\ell$ to obtain real-valued polynomials in Cartesian coordinates,
\begin{equation*}
\begin{split}
Z_{\ell m}(\vect{x}) &:= r^\ell \sum_{m'=\ell,\ell-1,\dots,-\ell} u_{\ell,mm'} Y_{\ell m'}\\
&\stackrel{!}{=} \sum_{p_1+p_2+p_3=\ell} c_{\ell m,\vect{p}} \cdot \vect{x}^{\vect{p}},
\end{split}
\end{equation*}
where $c_{\ell m,\vect{p}} \in \R$, $\vect{x}^{\vect{q}} := \prod_{i=1}^3 x_i^{q_i}$. Pluggin this into
Eq.~\eqref{eq:Orbitals} results in real-valued Slater-type orbitals given by
\begin{equation}
\psi_{n \ell m}(\vect{x}) = s_{n\ell}\,Z_{\ell m}(\vect{x}) \left(\sum_{i=0}^{n-\ell-1} d_{n\ell,i}\,r^i\right) \mathrm{e}^{-\frac{Z_{n\ell}}{n} r},
\label{eq:OrbitalsRealCartesian}
\end{equation}
where we have set $d_{n\ell,i} := b_{n\ell,i}\,c_{n\ell,i}$ to shorten notation.

Concretely, for $\ell = 1$ we adapt Ref.~\onlinecite{NuclearChargeLimit2009} and choose (in this order)
\begin{equation*}
\left(Z_{1 m}(\vect{x})\right)_m = \left(\text{pz},\text{px},\text{py}\right) := \frac{1}{2}\sqrt{\frac{3}{\pi}}\left(x_3,x_1,x_2\right).
\end{equation*}
For $\ell = 2$,
\begin{equation*}
\begin{split}
&\left(Z_{2 m}(\vect{x})\right)_m
= \left(\text{d0},\text{dz},\text{dm},\text{dx},\text{dy}\right) := \frac{1}{4}\sqrt{\frac{15}{\pi}}\\
&\times \left(\frac{2 x_3^2 - x_1^2 - x_2^2}{\sqrt{3}},\, 2 x_1 x_2,\, x_1^2-x_2^2,\, 2 x_2 x_3,\, 2 x_1 x_3\right),
\end{split}
\end{equation*}
or --- figuratively $\text{dx} \sim 2\,y\,z$, $\text{dy} \sim 2\,x\,z$, $\text{dz} \sim 2\,x\,y$, $\text{d0} \sim \left(3\,z^2 - r^2\right)/\sqrt{3}$, and $\text{dm} \sim x^2-y^2$. The corresponding unitary $U_\ell$ read
\begin{equation*}
U_1 = \frac{1}{\sqrt{2}}\left(\begin{smallmatrix}
 0 & \sqrt{2} & 0 \\
 -1 & 0 & 1 \\
 \ii & 0 & \ii
\end{smallmatrix}\right)
\quad\text{and}\quad
U_2 = \frac{1}{\sqrt{2}}\left(\begin{smallmatrix}
 0 & 0 & \sqrt{2} & 0 & 0 \\
 -\ii & 0 & 0 & 0 & \ii \\
 1 & 0 & 0 & 0 & 1 \\
 0 & \ii & 0 & \ii & 0 \\
 0 & -1 & 0 & 1 & 0
\end{smallmatrix}\right)
\end{equation*}
when arranging the spherical harmonics $Y_{\ell m}$ with decreasing quantum number $m$.


\subsubsection{Transformation to Fourier space}

We adapt the idea in Ref.~\onlinecite{NuclearChargeLimit2009} to calculate Coulomb integrals between pairs of spatial orbitals 
by applying Fourier transformation. We use the normalization-factor-free convention
\begin{equation*}
\left(\mathcal{F} f\right)(k) := \int_{\R^n} f(x) \mathrm{e}^{-\ii k \cdot x} \ud x.
\end{equation*}
Given one-electron orbitals $\varphi_1,\varphi_2,...$ with $\varphi_i$ and $\mathcal{F}\varphi_i \in L^2(\R^3) 
\cap L^\infty(\R^3)$, let $f(\vect{x}) := \varphi_i(\vect{x})\,\conj{\varphi_j(\vect{x})}$ and 
$g(\vect{x}) := \conj{\varphi_k(\vect{x})}\,\varphi_\ell(\vect{x})$. Then (see e.g. Ref.~\onlinecite{NuclearChargeLimit2009})
\begin{equation*}
\bprod{\varphi_i \varphi_j}{\varphi_k \varphi_\ell} = \frac{1}{2 \pi^2} \int_{\R^3} \frac{1}{\abs{\vect k}^2} \conj{\left(\mathcal{F} f\right)(\vect{k})} \left(\mathcal{F} g\right)(\vect{k}) \ud^3 \vect{k}.
\end{equation*}
Since we have switched to real-valued Cartesian orbitals in the previous subsection, $f(\vect{x}) = \varphi_i(\vect{x})\,\conj{\varphi_j(\vect{x})}$ can be expanded as
\begin{equation}
\label{eq:cnuq}
f(\vect{x}) = \sum_{\nu=0}^{\nu_{\max}} r^{\nu} \left(\sum_{q_1,q_2,q_3=0}^{q_{\max}} c_{\nu,\vect{q}} \cdot \vect{x}^{\vect{q}}\right)\mathrm{e}^{-\lambda r}, \quad r = \abs{\vect{x}}
\end{equation}
with constants $c_{\nu,\vect{q}}$ and $\lambda > 0$. Directly from the definition of the Fourier transformation, it follows that
\begin{equation}
\left(\mathcal{F} f\right)(\vect{k}) = \sum_{\nu,\vect{q}} c_{\nu,\vect{q}}\,(-1)^\nu \frac{\partial^\nu}{\partial\lambda^\nu}\, \ii^{q_1+q_2+q_3} \frac{\partial^\vect{q}}{\partial\vect{k}^\vect{q}} \left(\mathcal{F} \mathrm{e}^{-\lambda r}\right)(\vect{k}),
\end{equation}
where we have used the notation
\begin{equation*}
\frac{\partial^\vect{q}}{\partial \vect{k}^\vect{q}} := \prod_{i=1}^3 \frac{\partial^{q_i}}{\partial k_i^{q_i}} \quad\text{for each}\quad \vect{q} \in \N_0^3.
\end{equation*}
It is well known that
\begin{equation*}
\left(\mathcal{F}\mathrm{e}^{-\lambda r}\right)(\vect{k}) = \frac{8 \lambda \pi}{\left(\lambda^2 + k^2\right)^2}, \quad k = \abs{\vect{k}}.
\end{equation*}
Thus, precomputing the following integral over polar coordinates
\begin{equation*}
\begin{split}
&I_{\vect{q},\vect{q}'}(\lambda,\lambda')
:= (-\ii)^{q_1+q_2+q_3}\ \ii^{q_1'+q_2'+q_3'} \frac{1}{2\pi^2}\\
&\times \int_0^\infty \int_0^\pi \int_0^{2 \pi}
\left(\frac{\partial^\vect{q}}{\partial\vect{k}^\vect{q}} \frac{8 \lambda \pi}{\left(\lambda^2 + k^2\right)^2}\right)\\
&\quad \left(\frac{\partial^{\vect{q}'}}{\partial\vect{k}^{\vect{q}'}} \frac{8 \lambda' \pi}{\left(\lambda'^2 + k^2\right)^2}\right)
\sin\vartheta\,\ud\varphi \ud\vartheta \ud k
\end{split}
\end{equation*}
we obtain for the spinless Coulomb integrals~\eqref{eq:veeSpintrace} with orbitals~\eqref{eq:Orbitals}
\begin{equation}
\begin{split}
\left(\hat{v}_{ee}\right)_{ij,k\ell} &= \bprod{\varphi_i \varphi_j}{\varphi_k \varphi_\ell} = \sum_{\nu,\nu'} \sum_{\vect{q},\vect{q}'} \,\conj{c_{\nu,\vect{q}}} \,c_{\nu',\vect{q}'} \,\cdot\\
&\quad (-1)^{\nu+\nu'} \frac{\partial^\nu}{\partial\lambda^\nu} \frac{\partial^{\nu'}}{\partial\lambda'^{\nu'}} I_{\vect{q},\vect{q}'}(\lambda,\lambda'),
\end{split}
\label{eq:veeFourier}
\end{equation}
with $c_{\nu,\vect{q}}$, $\lambda$ as in Eq.~\eqref{eq:cnuq} and $c_{\nu',\vect{q}}$, $\lambda'$ the analogous constants for $\varphi_k(x)\overline{\varphi_\ell(x)}$. 

\subsubsection{Application to dilated Slater-type orbitals}
\label{sec:AppCoulombDilOrb}

Taking pairwise products of the wavefunctions~\eqref{eq:OrbitalsRealCartesian} involves the convolution of coefficients,
\begin{equation*}
\begin{split}
f(\vect{x})
&:= \psi_{n \ell m}(\vect{x})\,\conj{\psi_{n'\ell'm'}(\vect{x})}\\
&= s_{n\ell} s_{n'\ell'}\,Z_{\ell m}(\vect{x}) Z_{\ell'm'}(\vect{x})\\
&\times \left(\sum_i \left(d_{n\ell}*d_{n'\ell'}\right)_i r^i\right) \mathrm{e}^{-\left(\frac{Z_{n\ell}}{n}+\frac{Z_{n'\ell'}}{n'}\right)r}
\end{split}
\label{eq:fPairWavefunctions}
\end{equation*}
with the discrete convolution
\begin{equation*}
\left(d_{n\ell}*d_{n'\ell'}\right)_i = \sum_k d_{n\ell,k}\,d_{n'\ell',i-k}.
\end{equation*}
Similar reasoning applies to the product $Z_{\ell m} Z_{\ell'm'}$,
\begin{equation*}
Z_{\ell m}(\vect{x}) Z_{\ell'm'}(\vect{x}) = \sum_{\abs{\vect p}_1=\ell+\ell'} \left(c_{\ell m}*c_{\ell'm'}\right)_{\vect{p}} \cdot \vect{x}^{\vect{p}}.
\end{equation*}
Let
\begin{equation*}
g(\vect{x}) := \conj{\psi_{\tilde{n}\tilde{\ell}\tilde{m}}(\vect{x})}\,\psi_{\tilde{n}'\tilde{\ell}'\tilde{m}'}(\vect{x})
\label{eq:gPairWavefunctions}
\end{equation*}
be another pairwise product of wavefunctions. Then the Coulomb integral of these pairs equals
\begin{equation}
\label{eq:CoulombFactorized}
\begin{split}
&\bprod{\psi_{n \ell m} \psi_{n'\ell'm'}}{\psi_{\tilde{n}\tilde{\ell}\tilde{m}} \psi_{\tilde{n}'\tilde{\ell}'\tilde{m}'}}
= \iint_{\R^3} \frac{\conj{f(\vect{x})} g(\vect{y})}{\abs{\vect{x}-\vect{y}}} \ud^3\vect{x} \ud^3\vect{y}\\
&\stackrel{\eqref{eq:veeFourier}}{=} s_{n\ell}\,s_{n'\ell'}\,s_{\tilde{n}\tilde{\ell}}\,s_{\tilde{n}'\tilde{\ell}'} \sum_i \left(d_{n\ell}*d_{n'\ell'}\right)_i\,\sum_j \left(d_{\tilde{n}\tilde{\ell}}*d_{\tilde{n}'\tilde{\ell}'}\right)_j\\
&\quad\times \sum_{\abs{\vect p}_1=\ell+\ell'} \left(c_{\ell m}*c_{\ell'm'}\right)_{\vect{p}} \sum_{\abs{\vect q}_1=\tilde{\ell}+\tilde{\ell}'} \left(c_{\tilde{\ell}\tilde{m}}*c_{\tilde{\ell}'\tilde{m}'}\right)_{\vect q}\\
&\quad\times (-1)^{i+j} \frac{\partial^i}{\partial\lambda^i} \frac{\partial^j}{\partial\mu^j} I_{\vect{p},\vect{q}}(\lambda,\mu)\evalat_{\lambda=\frac{Z_{n\ell}}{n}+\frac{Z_{n'\ell'}}{n'},\mu=\frac{Z_{\tilde{n}\tilde{\ell}}}{\tilde{n}}+\frac{Z_{\tilde{n}'\tilde{\ell}'}}{\tilde{n}'}}.
\end{split}
\end{equation}

\section{Computing the CI levels and states}
\label{sec:Assembly}

Our overall algorithm for the CI method in Section~\ref{sec:AsymptoticsBased} consists of a symbolic part,
symmetry reduction and reduction to two-body space, and a numerical part, Hamiltonian matrix diagonalization and orbital exponent
optimization.

\subsection{Symbolic precomputation}

The following pre-computational steps will allow us to calculate the matrix representation of the 
Hamiltonian quickly, given plug-in values for the dilation parameters $Z_{n\ell}$.

\begin{enumerate}

\item Compute the simultaneous eigenspaces of the operators \eqref{eq:simLS}, via the algorithm described in 
Section~\ref{sec:LSDiag}. 

\item For any simultaneous eigenspace of the operators~\eqref{eq:simLS}, choose an orthonormal basis $(\psi_1,\dots,\psi_r)$ and calculate the one- and two-particle reduced density matrices $\gamma_{\ket{\psi_i}\bra{\psi_j}}$ and $\Gamma_{\ket{\psi_i}\bra{\psi_j}}$, respectively, of the $N$-particle states $\ket{\psi_i}\bra{\psi_j}$ for all $i,j=1,\dots,r$. Subsequently, trace out the spin part as defined in~\eqref{eq:Gamma2Spintrace} and~\eqref{eq:Gamma1Spintrace} to obtain $\hat{\gamma}_{\ket{\psi_i}\bra{\psi_j}}$ and $\hat{\Gamma}_{\ket{\psi_i}\bra{\psi_j}}$.

\item For the Slater orbitals~\eqref{eq:Orbitals}, calculate symbolic versions of the orthonormalization constants 
$c_{n\ell}$ and $s_{n\ell}$ in Section~\ref{sec:Orthonormalization} via equations~\eqref{eq:IterOrthogonal} 
and~\eqref{eq:OrbsNorm}, respectively. Note that these constants still depend on the dilation parameters $Z_{n\ell}$, 
which will be plugged in at the numerical optimization step below.

\item Calculate symbolic matrix representations \eqref{eq:h0Spintrace} and \eqref{eq:veeSpintrace}
of the single-particle and electron-interaction Hamiltonians $h_0$ and $v_{ee}$, using a computer algebra system
and Eq.~\eqref{eq:OneBodyInts} for $h_0$ and Eq.~\eqref{eq:CoulombFactorized} for $v_{ee}$.
These matrices still depend on the orthonormalization constants $s_{n\ell}$ and $c_{n\ell,i}$ 
from Step 2, and on the dilation parameters $Z_{n\ell}$.
\end{enumerate}

\subsection{Numerical diagonalization and energy minimization}

For any given set of orbital exponents $Z_{n\ell}$, we can now calculate and diagonalize
the matrix representation of the Hamiltonian projected onto any LS-eigenspace, by using the reduced density 
matrix formalism in section~\ref{sec:RDMs}. In mathematical terms,
\begin{enumerate}

\item For a current numerical value of the orbital exponents $Z_{1,0},\, Z_{2,0}, \, Z_{2,1}, ...$, evaluate
the symbolic orthonormalization constants $s_{n\ell}$ and $c_{n\ell,i}$ and the symbolic matrix elements of $\hat{h}_0$ and 
$\hat{v}_{ee}$. 

\item Equation~\eqref{eq:EvalSingle} yields the matrix elements of the Hamiltonian on an LS-eigenspace with
orthonormal basis $(\psi_1,\dots,\psi_r)$, namely,
\begin{equation*}
\hprod{\psi_i}{H\,\vert\,\psi_j} = \trace\left(\hat{h}_0 \hat{\gamma}_{\ket{\psi_j}\bra{\psi_i}}\right) + 
\trace\left(\hat{v}_{ee}\,\hat{\Gamma}_{\ket{\psi_j}\bra{\psi_i}}\right).
\end{equation*}
(Note that it would be theoretically possible but computationally inefficient to carry out this step symbolically.)

\item Obtain the ground state energy $E = \lambda_{\min}(\hprod{\psi_i}{H\,\vert\,\psi_j}_{i,j=1,\dots,r})$

\item Iteratively repeat these steps for different values of the orbital exponents within a suitable optimization routine to minimize 
the ground state energy numerically. (We used a gradient-free simplex search method.)

\end{enumerate}

\section{Cost analysis}
\label{sec:Costs}

In what follows we review the computational speedup of the central algorithmic steps as
compared to operating directly on the full $N$-particle Hilbert space $\wedge^N \mathcal{H}$.

\subsection{Configurations}
\label{sec:CostsConfigurations}

In this paragraph, we quantify the savings by the configuration calculus introduced in
Section~\ref{sec:ReductionConfigurations}. To shorten notation, set $g_j := \dim\!V_{u_j}$,
and assume that the total particle number $N$ is fixed. Thus, the dimension of the full
$N$-particle Hilbert space equals $\binom{\sum g_j}{N}$. Consider configurations
$\mathcal{C}^{n_1,\dots,n_k}$ with $\sum n_j = N$. They partition the Hilbert space,
and accordingly
\begin{equation*}
\begin{split}
\sum_{\substack{n_1,\dots,n_k\\\sum n_j=N}} \dim\left(\mathcal{C}^{n_1,\dots,n_k}\right)
&= \sum_{\substack{n_1,\dots,n_k\\\sum n_j=N}} \prod_j \binom{g_j}{n_j}\\
= \binom{\sum g_j}{N}
&= \dim\left(\wedge^N \mathcal{H}\right)
\end{split}
\end{equation*}
as expected. Now, assume we are given an algorithm of order $\mathcal{O}(\dim^p)$, like, e.g.,
LS diagonalization with $p = 3$. Running this algorithm either applied to all configurations
separately or to the full $N$-particle Hilbert space incurs computational costs of order
\begin{equation}
\sum_{\substack{n_1,\dots,n_k\\\sum n_j=N}}
\dim\left(\mathcal{C}^{n_1,\dots,n_k}\right)^p
\quad \text{as compared to}\quad \dim\left(\wedge^N \mathcal{H}\right)^p.
\label{eq:CostCompareConfigurations}
\end{equation}
In what follows, we derive an estimate of the quotient of these two terms. The
Stirling approximation of factorials and a logarithmic series expansion leads to
\begin{equation*}
\binom{g}{n} \approx 2^{\frac{1}{2}+g} (\pi g)^{-\frac{1}{2}} e^{-\frac{(g-2 n)^2}{2 g}}.
\end{equation*}
Plugging this into the left hand side of~\eqref{eq:CostCompareConfigurations} yields
\begin{equation*}
\begin{split}
&\sum_{\substack{n_1,\dots,n_k\\\sum n_j=N}} \prod_j \binom{g_j}{n_j}^p \approx \idotsint_{-\infty}^\infty \delta\left(N-\sum n_j\right)\\
&\quad \times \prod_j 2^{p\left(\frac{1}{2}+g_j\right)} \left(\pi g_j\right)^{-\frac{p}{2}} e^{-p \frac{(g_j-2 n_j)^2}{2 g_j}} \ud n_1 \cdots \ud n_k.
\end{split}
\end{equation*}
The Fourier transform of these integrals is the pointwise product of the individual Fourier transforms. One obtains
\begin{equation*}
\begin{split}
&\int_{-\infty}^\infty 2^{p \left(\frac{1}{2}+g\right)} \left(\pi g\right)^{-\frac{p}{2}} e^{-p\frac{(g-2 n)^2}{2 g}} e^{-\ii n t} \ud n\\
&\qquad = \left(2/\pi\right)^{\frac{1}{2}(p-1)} p^{-\frac{1}{2}} 2^{p g} g^{\frac{1}{2}(1-p)} e^{-g t (4 \ii p+t)/(8 p)}
\end{split}
\end{equation*}
for each individual transform. Now, the inverse Fourier transform of the pointwise products gives the desired approximation of the left hand side~\eqref{eq:CostCompareConfigurations}, namely
\begin{equation}
\begin{split}
&\left(2/\pi\right)^{\frac{1}{2}(1-k+k p)} p^{\frac{1}{2}(1-k)} \left(g_{\text{prod}}\right)^{\frac{1}{2}(1-p)}\\
&\times 2^{p\,g_{\text{sum}}} \left(g_{\text{sum}}\right)^{-\frac{1}{2}} e^{-p\frac{\left(g_{\text{sum}}-2 N\right)^2}{2 g_{\text{sum}}}}
\label{eq:CostConfigApprox}
\end{split}
\end{equation}
where we have set $g_{\text{sum}} := \sum_j g_j$ and $g_{\text{prod}} := \prod_j g_j$ to shorten notation. Finally, dividing the (Stirling approximated) right hand side of~\eqref{eq:CostCompareConfigurations} by~\eqref{eq:CostConfigApprox} yields the sought-after quotient
\begin{equation*}
\left(\frac{\pi}{2}\right)^{\frac{1}{2}(k-1)(p-1)} p^{\frac{1}{2}(k-1)} \left(\frac{g_{\text{prod}}}{g_{\text{sum}}}\right)^{\frac{1}{2}(p-1)}.
\end{equation*}
Note that this factor is independent of the particle number $N$. It equals $1$ for $p=1$, as expected.

As concrete example, consider Chromium with three active subshells $3p, 3d, 4s$, i.e., all subshells up to $3s$ are completely filled. Thus, in terms of the computation parameters we have $(g_1,g_2,g_3) \equiv (\dim\!V_p, \dim\!V_d, \dim\!V_s) = (6, 10, 2)$, $N_\text{eff} = 12$ (electron number in active orbitals) and algorithmic order $p = 3$, say. Then, the approximated quotient equals $5 \pi^2 \doteq 49.348$, which is close to the exact number $50774322144/938076521 \doteq 54.1$.

\subsection{LS diagonalization for sparse vectors}

In this paragraph, we show that the cost of the decomposition~\eqref{eq:SimDiagVI} essentially scales linearly in the problem size $\dim(V_I)$, assuming a sparse structure of the associated coefficient vectors.

First, consider two irreducible angular momentum eigenspaces $V_1$ and $V_2$ with quantum numbers $\ell_j$ and dimensions $(2\ell_j+1)$ ($j=1,2$, without loss of generality $\ell_1 \ge \ell_2$). Then the Clebsch-Gordan method partitions $V_1 \otimes V_2$ into total angular momentum eigenstates, i.e.,
\begin{equation*}
V_1 \otimes V_2 = \bigoplus_{\ell=\abs{\ell_1-\ell_2}}^{\ell_1+\ell_2} V_{12,\ell}, \quad \dim\left(V_{12,\ell}\right) = 2\ell+1.
\end{equation*}
Each $V_{12,\ell}$ requires the computation of exactly
\begin{equation}
\label{eq:costCG_V12ell}
\begin{split}
&\text{num}_{\text{CG}}\left(V_{12,\ell}\right)\\
&= \left(\ell_1+\ell_2+1\right)\left(2\ell+1\right) - \left(\ell_1-\ell_2\right)^2 - \ell(\ell+1)
\end{split}
\end{equation}
Clebsch-Gordan coefficients and Kronecker products $\varphi_1\otimes\varphi_2$ ($\varphi_j \in V_j$). Due to the mentioned sparse structure, we assume $\mathcal{O}(1)$ cost for each of these Kronecker products. Summing up~\eqref{eq:costCG_V12ell} for all $\ell$ yields
\begin{equation}
\begin{split}
&\sum_{\ell=\abs{\ell_1-\ell_2}}^{\ell_1+\ell_2} \text{num}_{\text{CG}}\left(V_{12,\ell}\right)
= \left(2\ell_2+1\right)\\
&\quad \times \left[\left(2\ell_1+1\right)\left(2\ell_2+1\right)-\frac{1}{3}\left(\left(2\ell_2+1\right)^2-1\right)\right]\\
&\le \left(2\ell_2+1\right) \dim\left(V_1 \otimes V_2\right).
\end{split}
\label{eq:costCG_V12}
\end{equation}
Now consider irreducible angular momentum eigenspaces $V_j$, $j=1,\dots,k$, with respective quantum numbers $\ell_j$. The computational cost of the iterated Clebsch-Gordan method will be dominated by the calculation of the total angular momentum eigenspaces of
\begin{equation*}
\left(\bigoplus_i V_{1,\dots,k-1;\tilde{\ell}_i}\right) \otimes V_k, \quad \dim\left(V_{1,\dots,k-1;\tilde{\ell}_i}\right) = 2\tilde{\ell}_i+1,
\end{equation*}
where each $V_{1,\dots,k-1;\tilde{\ell}_i}$ is an irreducible angular momemtum eigenspace in 
$\bigotimes_{j=1}^{k-1} V_j$ such that $\sum_i (2\tilde{\ell}_i+1) = \dim(V_1 \otimes\cdots\otimes V_{k-1})$. 
According to~\eqref{eq:costCG_V12}, this requires not more than $(2\ell_k+1) \cdot \dim(V_1 \otimes\cdots\otimes V_k)$ 
Clebsch-Gordan coefficients and associated Kronecker products. Thus, in case of all $V_j$ being of uniformly bounded dimension,
i.e., $\ell_1,\dots,\ell_k\le \ell_{\max}$, the cost is of order
\begin{equation}
\text{cost}_{\text{CG}}\left(V_1,\dots,V_k\right) = \mathcal{O}\left(\dim\left(V_1 \otimes\cdots\otimes V_k\right)\right).
\label{eq:costCG}
\end{equation}
So indeed, the cost is (almost) of the order of the problem size.

The analysis for spin states is exactly the same, and the angular momentum and spin operators can be treated independently. Thus, the result~\eqref{eq:costCG} remains valid when considering both angular momentum and spin.

\subsection{Diagonalization of the Hamiltonian}

We now consider the exact reduction steps introduced in subsections~\ref{sec:Symmetries} and~\ref{sec:RepresentationTheory}:
the Hamiltonian can be diagonalized within each LS eigenstate separately, and only states with quantum numbers $m_\ell \equiv 0, m_s \equiv s$ 
need to be taken into account. (Partitioning into configurations is advantageous for the LS diagonalization only, since the Hamiltonian
mixes configurations.) The latter saves a factor of~$(2\ell+1)\cdot(2 s+1)$ states with each $\vect{L}^2$-$\vect{S}^2$ eigenspace.
The former, in the examples in Section~\ref{sec:Applications}, reduces the number of states by a factor of $10^2$ to $10^3$. 

We illustrate the huge cost reduction by the example of the Chromium $^7S$ states with configurations [Ar]\,$3d^j 4s\,4p^k 4d^\ell$ 
such that $j+k+\ell = 5$ (see Section~\ref{sec:Applications}). The dimension of the full CI state space equals $\binom{26}{5} = 65780$ (since 5 valence electrons have to be allocated to 10+6+10 possible orbitals). By contrast, restricting to a typical symmetry subspace of interest, such as $^7S$ (i.e., $\vect{L}^2=0$ and $\vect{S}^2 = 3(3+1)$), reduces the dimension to 98, and taking $S_z$ maximal and $L_z=0$ reduces it further to 14.

\subsection{Storing RDMs instead of N-particle wavefunctions}
\label{sec:StorageSaving}

Even though the number of required wavefunctions has been reduced, each individual $N$-electron wavefunction on a $K$-orbital space still requires, a priori, $\binom{K}{N}$ entries.

First --- as illustrated in Section~\ref{sec:RDMs} --- this cost can be reduced since the components of the
Hamiltonian matrix on a given $N$-particle subspace only requires knowledge of the two-particle density
matrices of any pair of $N$-particle basis functions. These RDMs have $\binom{K}{2}^2 = \mathcal{O}(K^4)$ entries, namely
$(\Gamma_{\ket{\psi}\bra{\chi}})_{ij,k\ell}$ with $1 \le i < j \le K$ and $1 \le k < \ell \le K$.

Second, applying the spinless density matrix defined in equation~\eqref{eq:Gamma2Spintrace} reduces the
number $K$ of single-particle orbitals by one half.

Third, we note that the density matrix typically exhibits a sparse structure, so we actually need far fewer entries. This is related to prior LS diagonalization on the $N$-particle Hilbert space. More precisely, the two-particle RDM of an $N$-particle $\vect{L}^2$--$\vect{S}^2$--$L_z$--$S_z$ eigenstate must commute with these symmetry operators on the two-particle space. Reconsider, for instance, the $^7S$ states of Chromium with configuration [Ar]\,$3d^j 4s\,4p^k 4d^\ell$, $j+k+\ell=5$. A general spinless RDM with orbitals up to $4d$ has $\binom{23 + 1}{2}^2 = 76176$ entries. By contrast, the $14^2=196$ RDM's of the $14$ $^7S$ states with $S_z$ maximal and $L_z=0$ turn out to have, on average, only $94.3$ nonzero entries, the maximum number of nonzero entries which occurs being $648$.

\section{Anomalous filling of 4s and 3d orbitals in transition metal atoms}
\label{sec:Applications}

We have applied the algorithmic framework reported above to the calculation of ground and excited states and levels in 3d
transition metal atoms. These continue to offer substantial computational challenges, due to the irregular filling of 4s versus
3d orbitals, strong correlations, and non-negligible relativistic effects. 

Previous computations have led to different results, depending on the level of theory used. Limitations of single-determinant
Hartree-Fock theory for these atoms are discussed in Ref.~\onlinecite{MelroseScerri}.
Multi-determinant Hartree–Fock (HF) energies for the experimental ground state configurations (but not for competing configurations) are given
in Ref.~\onlinecite{FroeseFischer77} (with the exception of Cr), Ref.~\onlinecite{TwoOpenShells2007} (only for atoms with anomalous filling such as Cr), 
and Ref.~\onlinecite{Bungeetal}. The interconfigurational ordering of $4s^1 3d^n$ versus $4s^2 3d^{n-1}$ is discussed in Ref.~\onlinecite{Kagawa}
for relativistic HF and in Ref.~\onlinecite{HarrisJones,Yanagisawa} for DFT. Among the transition metal series Sc, Ti, V, Cr, Mn, Fe, Co, Ni,
Cu, relativistic HF rendered $4s^1$ stable for Cr, Mn, Fe, Ni, Cu, even though experimentally only Cr and Cu have
a $4s^1$ ground state.\footnote{In fact, for Ni the experimental classification as $4s^2$ should be viewed with some caution. A look
at the actual data\cite{NIST} shows that for Ni, relativistic $J$ splittings are of the same order as the interconfigurational gap, 
and while the experimental ground state is a particular $J$ state of the $4s^2$ ($^3F$) configuration, $4s^1$ ($^3D$) would become stable 
if one averages over $J$ according to multiplicity.} DFT does not fare better, regardless of the type of exchange-correlation
functional used: $4s^1$ is rendered stable by Becke 88 for Ti, V, Cr, Ni, Cu\cite{Yanagisawa}, by the local density approximation 
and Perdew-Wang for V, Cr, Co, Ni, Cu\cite{HarrisJones,Yanagisawa}, and by B3LYP for V, Cr, Co, Ni, Cu\cite{Yanagisawa}.   
The poor atomization energies of DFT functionals such as Becke 88 and B3LYP for transition metal dimers (those 
for Cr$_2$ even come out with the wrong sign) have been associated\cite{Yanagisawa} to poor 
interconfigurational energies of the atoms. It is then of interest to revisit the latter from alternative theoretical points of view. 

Our results for the asymptotics-based CI model (A), (B), (C) in Section~\ref{sec:AsymptoticsBased} are as follows. First,
we considered a minimal model for the third period elements K to Zn with configurations $[Ar]3d^j 4s^k$, that is to say
in the language of Section~\ref{sec:AsymptoticsBased} we choose the cutoffs
\begin{equation*}
(n,\ell)_{\min} = (3,1) = 3p, \qquad (n,\ell)_{\max} = (4,0) = 4s.
\end{equation*}
It turns out that the ground states from Ca to Zn always put two electrons in the 4s subshell, i.e. have configuraton
$[Ar]3d^j 4s^2$. See Table~\ref{tab:3d4s}. Thus minimal asymptotics-based CI coincides with the empirical Madelung rule
(which states that the subshells are filled in the order of increasing $n+\ell$ and, for equal $n+\ell$, in the order of 
decreasing $n$). Experimentally, this means that the method fails for the two anomalous atoms Cr and Cu.

\begin{table}
\centering
\begin{tabular}{|c||c|c|c|c|c|c|}
\hline
Atom&\multicolumn{2}{c|}{Sym}&dim&\multicolumn{3}{c|}{energy [a.u.]} $\vphantom{\Bigl(}$\\
&CI&exp&(subsp)&CI&exp&MDHF\\
\hline
 K&$^2S$&$^2S$&1   & -596.7993  & -601.9337 & -599.16478$\vphantom{\Bigl(}$\\
Ca&$^1S$&$^1S$&2   & -674.2442  & -680.1920 & -676.75818$\vphantom{\Bigl(}$\\
Sc&$^2D$&$^2D$&4   & -756.8908  & -763.8673 & -759.73571$\vphantom{\Bigl(}$\\
Ti&$^3F$&$^3F$&5   & -845.1599  & -853.3503 & -848.40599$\vphantom{\Bigl(}$\\
 V&$^4F$&$^4F$&4   & -939.1657  & -948.8394 & -942.88433$\vphantom{\Bigl(}$\\
Cr&$\mathbf{^5D}$&$^7S$&3   &-1039.0409 &-1050.4914 & -1043.3563$\vphantom{\Bigl(}$\\
Mn&$^6S$&$^6S$&1   &-1144.9715  &-1158.2670 &-1149.8662$\vphantom{\Bigl(}$\\
Fe&$^5D$&$^5D$&1   &-1256.7813  &-1271.6930 &-1262.4436$\vphantom{\Bigl(}$\\
Co&$^4F$&$^4F$&2   &-1374.8903  &-1393.3526 &-1381.4145$\vphantom{\Bigl(}$\\
Ni&$^3F$&$^3F$&1   &-1499.3759  &-1520.6907 &-1506.8709$\vphantom{\Bigl(}$\\
Cu&$\mathbf{^2D}$&$^2S$&1   &-1630.3692   &-1655.1317 &-1638.9637$\vphantom{\Bigl(}$\\
Zn&$^1S$&$^1S$&1   &-1768.0729&--- & -1777.8481 $\vphantom{\Bigl(}$\\
\hline
\end{tabular}
\caption{Ground state symmetries and energies from K to Zn
predicted by minimal asymptotics-based CI (this paper) with active space $[Ar]\,3d^j\,4s^k$
and compared to experimental data\cite{LideCRC2003}. Boldface denotes
deviation from experiment. The dimension of the joint eigenspace of the symmetry operators~\eqref{eq:simLS} which contains the unique ground state with $L_z=0$, $S_z$ maximal is denoted by dim.
Also shown are multi-determinant Hartree-Fock energies for the experimental ground state symmetries\cite{Bungeetal}
(for Ti and Cr see also Ref.~\onlinecite{SeveralOpenShells2005,TwoOpenShells2007}).}
\label{tab:3d4s}
\end{table}

Next, to address this issue we enlarged the CI subspace for the series Ca, Sc, Ti, V, Cr by the higher subshells
$4p$ and $4d$. That is to say we changed the cutoffs to 
\begin{equation*}
(n,\ell)_{\min} = (3,1) = 3p, \qquad (n,\ell)_{\max} = (4,2) = 4d,
\end{equation*}
hence including all configurations [Ar]$4d^j 4s^k 4p^\ell 4d^m$, and restricted to $k=1$ ($4s^1$) and 
$k=2$ ($4s^2$), respectively. In each case, we considered only the L and S values
selected by Hund's rules (i.e., we minimized first $S$ and then $L$, taking into account
one $s$ and $d$ subshell as in the minimal model above), computed the corresponding symmetry
subspaces via the algorithm in Section~\ref{sec:LSDiag}, and determined the associated eigenstates
and energy levels. The results are shown in Table~\ref{tab:GroundstateEnergies}.

\begin{table*}[!ht]
\centering
\begin{tabular}{|c||c|c|c|c||c|c|c|c|c|c|c|c|c|}
\hline
&Sym&Config&dim&$E_\text{CI}$ [a.u.]&$Z_{1s}$&$Z_{2s}$&$Z_{2p}$&$Z_{3s}$&$Z_{3p}$&$Z_{3d}$&$Z_{4s}$&$Z_{4p}$&$Z_{4d}$ $\vphantom{\Bigl(}$\\
\hline\hline
Ca&$^3D$&$4s^1$&2&-674.1634&19.68&17.41&16.13&12.05&10.38&2.83&5.43&-&2.46 $\vphantom{\Bigl(}$\\
&$\mathbf{^1S}$&$4s^2$&1&\textit{-674.2442}&19.68&17.41&16.13&12.10&10.38&-&5.03&-&- $\vphantom{\Bigl(}$\\
\hline
Sc&$^4F$&$4s^1$&3&-756.9381&20.68&18.42&17.15&12.99&11.30&8.26&5.35&-&6.24 $\vphantom{\Bigl(}$\\
&$\mathbf{^2D}$&$4s^2$&2&\textit{-756.9968}&20.68&18.42&17.15&13.06&11.34&10.07&5.31&-&8.46 $\vphantom{\Bigl(}$\\
\hline
Ti&$^5F$&$4s^1$&8&-845.3714&21.68&19.43&18.16&13.89&12.18&9.91&5.51&1.45&7.75 $\vphantom{\Bigl(}$\\
&$\mathbf{^3F}$&$4s^2$&3&\textit{-845.4210}&21.68&19.43&18.16&13.98&12.23&11.30&5.52&-&9.67 $\vphantom{\Bigl(}$\\
\hline
V&$^6D$&$4s^1$&17&-939.5952&22.68&20.44&19.17&14.78&13.04&11.20&5.61&1.88&8.93 $\vphantom{\Bigl(}$\\
&$\mathbf{^4F}$&$4s^2$&8&\textit{-939.6375}&22.68&20.44&19.17&14.86&13.10&12.36&5.70&5.25&10.62 $\vphantom{\Bigl(}$\\
\hline
Cr&$\mathbf{^7S}$&$4s^1$&14&\textit{-1039.7864}&23.68&21.44&20.18&15.64&13.89&12.37&5.67&9.51&10.00 $\vphantom{\Bigl(}$\\
&$^5D$&$4s^2$&17&-1039.7852&23.68&21.44&20.18&15.74&13.95&13.36&5.87&0.93&11.49 $\vphantom{\Bigl(}$\\
\hline
\end{tabular}
\caption{Asymptotics-based CI results with active space [Ar]\,$3d^j 4s^1 4p^k 4d^\ell$ and [Ar]\,$3d^j 4s^2 4p^k 4d^\ell$, respectively (this paper);
boldface denotes the experimental ground state symmetry, and italic font the lower of each pair of calculated energies, 
in exact agreement with the experimental data. Fourth column: Dimension of the symmetry subspace
containing the ground state.}
\label{tab:GroundstateEnergies}
\end{table*}

Despite the smallness of the radial basis set, the predicted ground state configurations and spin and angular momentum quantum numbers
are in full agreement with the experimental data. Physically, interesting insights can be gained from the orbital exponents
in Table~\ref{tab:GroundstateEnergies}, and from the coefficients of the different configurations contained in the ground 
state. First, for Ca, the $4s$ electron is more tightly bound than any $d$ electrons, whereas for Sc, Ti, V, Cr, 
this effect is reversed, in both the $4s^1$ and the $4s^2$ configuration, with $4s$ outside of both $3d$ and $4d$.
Second, considering for instance the $4s^1$ ($^7S$) Cr ground state, the configurations and weight coefficients
of the fourteen contributing basis states spanning the $^7S$, $m_L=0$, $m_S=3$ symmetry subspace of
$3d^j 4s^1 4p^\ell 4d^m$ are as follows:

\begin{small}
\begin{tabular}{l@{$\quad$}r@{.}l@{$ $}l}
$3d^5 4p^0 4d^0$     &0&36&\\
$3d^4 4p^0 4d^1$     &0&63&\\
$3d^3 4p^2 4d^0$     &0&056&\\
$3d^3 4p^0 4d^2$ (2D)&0&31 & and 0.50\\
$3d^2 4p^2 4d^1$ (2D)&0&036& and 0.038\\
$3d^2 4p^0 4d^3$ (2D)&0&17 & and 0.28\\
$3d^1 4p^2 4d^2$ (2D)&0&016& and 0.014\\
$3d^1 4p^0 4d^4$     &0&096&\\
$3d^0 4p^2 4d^3$     &0&0036&\\
$3d^0 4p^0 4d^5$     &0&012&\\
\end{tabular}
\end{small}

\noindent
In particular, no configuration dominates, and the highest weight configuration is not the naively expected $3d^5$
which one would enforce in both single-determinant HF and (L-S-adapted) multi-determinant HF, 
but $3d^44d^1$ (weight 0.63), followed by $3d^34d^2$ (0.59), $3d^5$ (0.36), and $3d^24d^3$ (0.33).
The highest-weight Cr $^7S$ basis function in which one of the $3d$ electrons has migrated to a $4d$ 
orbital is
\begin{equation*}
\begin{split}
(&\ket{3d_2 3d_1 3d_0 3d_{n\!1} 4s 4d_{n\!2}} - \ket{3d_2 3d_1 3d_0 3d_{n\!2} 4s 4d_{n\!1}}\\
+ &\ket{3d_2 3d_1 3d_{n\!1} 3d_{n\!2} 4s 4d_0} - \ket{3d_2 3d_0 3d_{n\!1} 3d_{n\!2} 4s 4d_1}\\
+ &\left.\ket{3d_1 3d_0 3d_{n\!1} 3d_{n\!2} 4s 4d_2}\right)/\sqrt{5},
\end{split}
\end{equation*}
with expressions of similar type for the remaining 13 basis functions.
Despite the simple treatment of radial orbitals here, our results
provide clear evidence of strong $3d$--$4d$ inter-shell correlations in Cr, and suggests (by comparing
energies of Tables~\ref{tab:3d4s} and~\ref{tab:GroundstateEnergies}) a huge, symmetry-reversing, 
correlation energy in Cr of the order of 1 a.u. 

For more quantitative conclusions the
radial basis set used here is too small, as is illustrated by our systematically higher energies compared
to the large-basis MDHF energies in Table~\ref{tab:3d4s}.
Our results constitute, however, an important step 
towards an accurate quantitative computation of the correlation energy. The remaining step, which
lies beyond the scope of the present paper, is to combine
the exact lowest symmetry subspaces delivered by our symmetry reduction algorithm with
high-accuracy, multi-parameter, self-consistent radial orbital optimization routines as 
have been developed for Hartree-Fock theory\cite{HJO00,Bungeetal,TwoOpenShells2007,NIST}.

\section{Conclusions}
\label{sec:Conclusions}

We have developed and implemented an algorithm for CI calculations for atoms which allows full resolution of
valence electron correlations in a large active space, via efficiently automated (and exact) symmetry reduction. Application to
$3d$ transition metal atoms shows that even very small radial basis sets yield the correct qualitative picture of
the electronic structure when all correlations within and between the $3d$, $4s$, $4p$, $4d$ shells 
are fully resolved and when orbital exponents are optimized self-consistently for the actual CI wavefunctions. 
We trace the qualitative accuracy of our results partly to the theoretical fact that the 
asymptotics-based CI method used here yields the correct leading-order asymptotics for the low-lying spectral gaps
in the fixed-$N$, large-$Z$ limit. 

In subsequent work, we aim to obtain an accurate quantitative picture, by combining the careful algorithmic treatment of correlations 
introduced here with suitable large-parameter orbital optimization routines as are used in (numerical or Roothaan-type)
Hartree-Fock theory.

%

\end{document}